\documentclass{article}[11pt]

\usepackage[fencedCode,inlineFootnotes,citations,definitionLists,hashEnumerators,smartEllipses,hybrid]{markdown}
\usepackage[utf8]{inputenc}
\usepackage[T1]{fontenc}
\usepackage[most]{tcolorbox}
\usepackage[scr=txupr]{mathalfa}
\usepackage{subcaption}
\usepackage{xcolor}
\usepackage[british]{babel}
\usepackage[autostyle]{csquotes}
\usepackage[numbers]{natbib}
\usepackage{mathtools}
\usepackage{amsmath, amsthm,amssymb,amstext,amsbsy} 
\usepackage{tabularx}
\usepackage{thmtools,thm-restate}
\usepackage{fullpage}
\usepackage{todonotes}[color=white]
\usepackage{enumitem}
\setlist[description]{leftmargin=15pt,labelindent=15pt}

\usepackage{comment}
\usepackage[bookmarks=false,pagebackref=true]{hyperref}  
\definecolor{darkgreen}{rgb}{0,0.5,0}
\hypersetup{linktocpage=true,colorlinks=true,citecolor=red,linkcolor=darkgreen} 

\usepackage{setspace}
\usepackage{cleveref}

\newcommand{\bone}{\mathbf{1}}
\newcommand{\bzero}{\mathbf{0}}
\newcommand{\bv}{\bar{v}}
\newcommand{\hv}{\hat{v}}
\newcommand{\Cov}{\mathsf{Cov}}
\newcommand{\CIsub}{\mathcal{I}_{\textsf{block}}}
\newcommand{\zl}[1]{z_{\text{long}, #1}}
\newcommand{\zs}[1]{z_{\text{short}, #1}}
\newcommand{\bmu}{\mathbf{\mu}}
\newcommand{\bCov}{\mathbf{\Sigma}}

\newif\ifnotes\notestrue

\ifnotes
 \usepackage{color}
 \definecolor{mygrey}{gray}{0.50}
 \newcommand{\notename}[2]{{\textcolor{red}{\footnotesize{\bf (#1:} {#2}{\bf ) }}}}

\else
 
 \newcommand{\notename}[2]{{}}
 
\fi

\newtheorem{theorem}{Theorem}[section]
\newtheorem{claim}[theorem]{Claim}

\newtheorem{proposition}[theorem]{Proposition}
\newtheorem{lemma}[theorem]{Lemma}
\newtheorem{corollary}[theorem]{Corollary}
\newtheorem{conjecture}[theorem]{Conjecture}
\newtheorem{remark}[theorem]{Remark}
\newtheorem{definition}[theorem]{Definition}
\newtheorem{open}[theorem]{Open Problem}

\newcommand{\R}{\mathbb{R}}

\newcommand{\BR}{\mathbb{R}}
\newcommand{\BE}{\mathbb{E}}
\newcommand{\BP}{\mathbb{P}}
\newcommand{\CP}{\mathfrak{P}}

\newcommand{\pmone}{\{\pm1\}}

\newcommand{\CI}{\mathcal{I}}
\newcommand{\CE}{\mathcal{E}}

\newcommand{\polylog}{\mathsf{polylog}}
\newcommand{\poly}{\mathsf{poly}}
\newcommand{\infnorm}[1]{\left\| #1 \right\|_{\infty}}
\newcommand{\eps}{\epsilon}
\newcommand{\disc}{\mathsf{Disc}}

\newcommand{\CS}{\mathcal{S}}

\newcommand{\CT}{\mathcal{T}}
\newcommand{\calT}{\mathcal{T}}

\newcommand{\p}{\mathbb{P}}

\newcommand{\herdisc}{\mathsf{HerDisc}}

\newcommand{\LP}{\mathsf{LP}}
\newcommand{\diag}{\mathsf{diag}}
\newcommand{\B}{\mathbb{B}}
\renewcommand{\S}{\mathbb{S}}
\newcommand{\predisc}{\mathsf{PrefixDisc}}
\newcommand{\prefix}{\mathsf{prefix}}
\newcommand{\pathdisc}{\mathsf{DAGDisc}}
\newcommand{\cdisc}{\mathsf{CombDisc}}
\renewcommand{\CT}{\calT}
\renewcommand{\CP}{\mathcal{P}}
\newcommand{\gauss}{\mu_d}

\makeatletter
\newcommand{\altqedhere}{%
  \ifmeasuring@\else\sbox0{\popQED}\fi
  \tag*{\qedsymbol}%
}
\makeatother

\newcommand{\haotian}[1]{{\color{blue} \textbf{Haotian:} #1}}

\newcommand{\ignore}[1]{{}}

\newtheorem{thm}{Theorem}[section]
\newtheoremstyle{TheoremNum}
        {\topsep}{\topsep}             
        {}                    
        {}                             
        {\bfseries}                  
        {}                            
        { }                           
        {\thmname{#1}\thmnote{ \bfseries #3}}
    \theoremstyle{TheoremNum}

\title{Prefix Discrepancy, Smoothed Analysis, \\ and  Combinatorial Vector Balancing}

\author{ Nikhil Bansal\thanks{University of Michigan, Ann Arbor, \texttt{bansal@gmail.com}. Supported in part by the  NWO VICI grant 639.023.812. The work was done when the author was at CWI, Amsterdam.} \and
 Haotian Jiang\thanks{Paul G. Allen School of CSE, University of Washington, Seattle, \texttt{jhtdavid@cs.washington.edu}.
Supported in part by the National Science Foundation, Grant Number CCF-1749609, CCF-1740551, DMS-1839116.
} \and
 Raghu Meka\thanks{Department of Computer Science, University of California, Los Angeles, \texttt{raghum@cs.ucla.edu}. Supported by NSF Grant CCF-1553605. 
} \and
Sahil Singla\thanks{Georgia Tech, \texttt{ssingla@gatech.edu}. The work was partly done when the author was at Princeton University.} \and 
 Makrand Sinha\thanks{Simons Institute and UC Berkeley, \texttt{makrand@berkeley.edu}. Work done while at CWI Amsterdam and supported by the NWO VICI grant 639.023.812.}
}
\date{}

\usepackage{tocloft}

\begin{document}

\maketitle
\pagenumbering{roman}
A well-known result of Banaszczyk in discrepancy theory concerns the \emph{prefix discrepancy} problem (also known as the {\em signed series} problem): given a sequence of $T$ unit vectors in $\R^d$,  find $\pm$ signs for each of them such that the signed sum vector along any prefix has a small $\ell_\infty$-norm? 
This problem is central to proving upper bounds for the Steinitz problem, and 
the popular  Koml\'os problem is a special case where one is only concerned with the final signed sum vector instead of all prefixes. 

Banaszczyk gave an $O(\sqrt{\log d+ \log T})$ bound for the  prefix discrepancy problem.  We investigate the tightness of Banaszczyk's bound and consider natural generalizations of prefix discrepancy:

\begin{itemize}
    \item We first consider a \emph{smoothed analysis} setting, where a small amount of additive noise perturbs the input vectors. We show an exponential improvement in $T$ compared to Banaszczyk's bound. Using a primal-dual approach and a careful chaining argument, we show that one can achieve a bound of $O(\sqrt{\log d+ \log\!\log T})$ with high probability in the smoothed setting.
    Moreover, this smoothed analysis bound is the best possible without further improvement on Banaszczyk's bound in the worst case.
    
    \item We also introduce a generalization of the prefix discrepancy problem to arbitrary DAGs. Here, vertices correspond to unit vectors, and the discrepancy constraints correspond to paths on a DAG on $T$ vertices --- prefix discrepancy is precisely captured when the DAG is a simple path. We show that an analog of Banaszczyk's $O(\sqrt{\log d+ \log T})$ bound continues to hold in this setting for adversarially given unit vectors and that the $\sqrt{\log T}$ factor is unavoidable for DAGs. We also show that unlike for prefix discrepancy, the dependence on $T$ cannot be improved significantly in the smoothed case for DAGs. 
    
    \item We conclude by exploring a more general notion of vector balancing, which we call \emph{combinatorial vector balancing}. In this problem, the discrepancy constraints are generalized from paths of a DAG to an arbitrary set system.  We obtain near-optimal bounds in this setting, up to poly-logarithmic factors.
\end{itemize}

\ignore{
A well-known result of Banaszczyk in discrepancy theory concerns the prefix discrepancy problem (also known as the signed series problem): given a sequence of $T$ unit vectors in $\R^d$, can one find $\pm$ signs such that signed sum vector along any prefix  has a small $\ell_\infty$-norm? 
\haotian{(Do we want to add the following sentence here?) The prefix discrepancy problem has major applications to the Steinitz problem.}
The popular  Koml\'os problem is a special case where one is only concerned with the final signed sum vector, instead of all prefixes. Banaszcyck gave an $O(\sqrt{\log d+ \log T})$ bound for the  prefix discrepancy problem. 

We investigate the tightness of Banaszcyck's bound and consider natural generalizations of prefix discrepancy:


\begin{itemize}
    \item We first consider a \emph{smoothed analysis} setting, where one is allowed to perturb the input vectors by adding a small noise and show that the dependence on $T$ in Banaszcyzk's bound can be exponentially improved. Using a primal-dual approach and chaining arguments, we show that one can achieve a bound of $O(\sqrt{\log d+ \log\!\log T})$ with high probability in the smoothed setting. 
    
    \item We  study a notion of prefix discrepancy for DAGs --- here vertices correspond to unit vectors and the discrepancy constraints correspond to paths on a DAG on $T$ vertices (this captures prefix discrepancy when the DAG is the directed path on $T$ vertices). We show that an analog of Banaszczyk's $O(\sqrt{\log d+ \log T})$ bound continues to hold in this setting for adversarially given unit vectors. We also show that the dependence on $T$ cannot be improved exponentially in the smoothed case for DAGs, unlike in the case of prefix discrepancy. 
    
    \item We conclude by exploring a more general notion of vector balancing, which we call \emph{combinatorial vector balancing}, where the discrepancy constraints correspond to an arbitrary set system. \haotian{Do we want to say the combinatorial vector balancing problem captures the prefix discrepancy problem and for DAGs?} We obtain near-optimal bounds in this setting, up to poly-logarithmic factors.
\end{itemize}
}

\setcounter{tocdepth}{1}

 
{\small
\begin{spacing}{0}
   \tableofcontents
\end{spacing}
}

\newpage

\pagenumbering{arabic}
\setcounter{page}{1}
\allowdisplaybreaks

\section{Introduction}
  

Given a sequence of $T$ vectors $v_1, \ldots, v_T$ in the Euclidean unit ball $\B_2^d$ (i.e., of $\ell_2$-length at most $1$), the \emph{prefix discrepancy} problem, also widely known as the \emph{signed series} problem, asks to find a coloring\footnote{A coloring is also referred to as a signing in the literature. We will use these two words interchangeably.}  $x \in \pmone^T$ to minimize $\max_{\tau \in [T]} \| \sum_{t\leq \tau} x_t v_t \|_\infty$. The classical Koml\'os problem is a special case of this problem where we want to minimize only the final discrepancy $\| \sum_{t\leq T} x_t v_t \|_\infty$. The prefix discrepancy problem was introduced by Spencer~\cite{Spencer77} who showed that there always exists a coloring with a discrepancy bound that only depends on $d$ (i.e., independent of $T$). This was later improved to $O(d)$ by Barany-Grinberg~\cite{BaranyGrinberg81}. Since then, the prefix-discrepancy problem has been greatly studied and has found several important applications. For instance, it implies a bound on the classical Steinitz problem~\cite{Chobanyan94,Steinitz-16} on the rearrangement of vector sequences, it appears in online discrepancy problems where one is interested in bounding the discrepancy at all times~\cite{BJSS20,BJMSS-SODA21,ALS-STOC21}, and it also implies the best known bound for Tusn\'ady's problem~\cite{Nikolov-Mathematika19}. Additionally, the Steinitz problem itself has many applications, including faster algorithms for solving integer programs~\cite{BMMP12,EW18,JR-ITCS19} and scheduling \cite{Barany81, Sevastjanov94}.

In a remarkable result, Banaszczyk~\cite{B12} showed that the prefix-discrepancy problem admits a coloring of $O(\sqrt{\log d + \log T})$ discrepancy, thereby exponentially improving the dependency on $d$ in the Barany-Grinberg bound, but incurring a $\sqrt{\log T}$ dependence on the number of vectors $T$. 

\begin{theorem}[\cite{B12}]
\label{thm:Banaszczyk12}
Given $v_1, \cdots, v_T \in \B_2^d$, there always exists a signing $x \in \pmone^T$ such that  
\[ \max_{\tau \in [T]} \Big\| \sum_{t\leq \tau} x_t v_t \Big\|_\infty = O(\sqrt{\log d+ \log T}) \enspace. \]
\end{theorem}
Given the many applications of prefix discrepancy, a natural question is whether we can improve this dependency on $T$ in \Cref{thm:Banaszczyk12} while still being poly-logarithmic in $d$. Optimistically, one could even hope for removing the dependence on $T$ altogether. Variants and special cases of this question have previously been asked~\cite{Spencer86Prefix,B12} 
and remain open (discussed in \Cref{sec:open}). Besides being a natural (and beautiful) question at the border of combinatorics and geometry, further motivation to study prefix discrepancy comes from the Koml\'os problem: there is currently no known separation between the prefix discrepancy and the Koml\'os problems\footnote{Note that for the Koml\'os problem one can always assume that $T\le d$ using a standard linear algebraic trick, while this might not be possible if one considers all prefixes of the signed sum vector.}. Thus, obtaining better upper or lower bounds for prefix discrepancy will also help us understand the avenues for making progress on the Koml\'os  conjecture.  

In this work, we strengthen Banaszczyk's result for the prefix discrepancy problem in two natural ways. First, we exponentially improve the dependence on $T$ in Banaszczyk's  prefix discrepancy bound in a \emph{smoothed setting}. Next, we extend the prefix discrepancy problem beyond \emph{prefixes of a path}, to \emph{prefixes of a given DAG}, as well as to more general combinatorial settings.

\subsection{Smoothed Analysis of Prefix Discrepancy}

Our first result concerns the prefix discrepancy problem in a \emph{smoothed analysis} setting. The motivation is that for several important problems, although the worst-case algorithmic bounds are bad, these worst-case instances could be brittle. 
The study of algorithmic problems in a smoothed model was initiated by Spielman and Teng~\cite{ST-JACM04}.  They showed that although the popular Simplex method can take an exponential time to solve an adversarially chosen linear programming instance, a slight perturbation of this instance allows the Simplex method to run in polynomial time. Since for many applications the input already contains some measurement noise, the analysis of algorithms for slight perturbations of the input is a natural model. Since the work of Spielman and Teng on linear programs, smoothed analysis has been successfully used for analyzing several other problems such as Local Search, Mechanism Design, and computing Pareto curves (see, e.g., the book~\cite{Roughgarden-Book20}).

In the smoothed setting of prefix discrepancy\footnote{A recent independent work has also introduced a smoothed model for  discrepancy minimization~\cite{HRS-arXiv21}. However, they consider an \emph{online} setting, and hence their model and results are incomparable.}, an adversary first chooses $T$ vectors $\bv_1, \ldots, \bv_T \in \B_2^d$ and then each vector $\bv_t$ is perturbed by a random independent noise $\hv_t \in \B_2^d$ with covariance $\Cov(\hv_t) \succcurlyeq \eps^2 I_d$. The goal of the algorithm is to solve the prefix discrepancy problem on input vectors $v_1, \ldots, v_T$ where $v_t := \bv_t + \hv_t$. In other words,  find a signing $x \in \pmone^T$ that minimizes 
\[ 
\predisc(v_1, \cdots, v_T) := \min_{x \in \pmone^T} \max_{\tau \in [T]}  \Big\|\sum_{t \le \tau} {x_t v_t} \Big\|_\infty . 
\] 
This smoothed setting also generalizes  stochastic discrepancy models where the input vectors are drawn i.i.d. from a given distribution, i.e., when we take all vectors $\bv_t=0$ and the noise vectors $\hv_t$ are drawn i.i.d.

In the above smoothed  setting, we exponentially improve the dependency on $T$ in Banaszczyk's bound.

\begin{restatable}[Smoothed analysis upper bound]{theorem}{smooth} \label{thm:smoothed_upper_bound}
 Suppose for each $t \in [T]$ the noise $\hv_t \in \B_2^d$ is sampled independently and satisfies covariance $\Cov(\hv_t) \succcurlyeq \eps^2 I_d$ for some $\epsilon \ge 1/{\poly(d,\log T)}$. Then, the prefix discrepancy in the smoothed setting is $O(\sqrt{\log d + \log\!\log T})$  with probability at least $1-1/\poly(T)$.
\end{restatable}

Note that the assumption on noise in \Cref{thm:smoothed_upper_bound} is quite mild and captures many natural perturbation models. 
For instance, the theorem holds even if each vector is \emph{unchanged} with probability $1-\epsilon$ for $\epsilon = 1/\poly(d, \log T)$ and we perturb a \emph{single random} coordinate by an additive $\pm \epsilon$ with probability $\epsilon$. 
 This is in contrast to the recent line of work that studies discrepancy minimization  for stochastic inputs~\cite{EzraL19,BansalMeka-SODA19,HobergRothvoss-SODA19,Potukuchi-ICALP20}, where the input stochasticity is used very heavily.
 
 \Cref{thm:smoothed_upper_bound} is also interesting since it suggests that if there is an adversarial $\Omega(\sqrt{\log d + \log T})$ lower bound instance for the prefix discrepancy problem then such an instance must be quite brittle: a slight perturbation of the instance will admit much better upper bounds. 
 Our proof of the theorem relies on a primal-dual approach and a careful chaining argument (see \Cref{sec:overview}). We note that it does not give us an efficient algorithm since it relies on \Cref{thm:Banaszczyk12}, which we do not know how to do algorithmically for prefixes~(see \Cref{sec:open}). However, one can get a slightly weaker  bound of $O(\log d + \log\!\log T)$ in polynomial time using a recent algorithm of~\cite{ALS21-arxiv} instead of \Cref{thm:Banaszczyk12}. 
 
 We also show that \Cref{thm:smoothed_upper_bound} is tight, unless there is a breakthrough for adversarial prefix discrepancy.
 
 
\begin{theorem}[Informal \Cref{thm:stochastic_lower_bound}]
The bound in \Cref{thm:smoothed_upper_bound} is the best possible for the smoothed setting assuming the tightness of Banaszczyk's $O(\sqrt{\log d + \log T})$ bound in the adversarial setting.
\end{theorem}



\subsection{Prefix Discrepancy for DAGs}


We next consider a substantial generalization of the prefix discrepancy problem to arbitrary directed acyclic graphs (DAGs). To motivate it, let us consider a slightly different view of the prefix discrepancy problem with input vectors $v_1,v_2,\ldots,v_T \in \B_2^d$. Consider the directed path, $\mathsf{Path}_T$, on $T$ vertices and suppose that we have a vector $v_t$ at every vertex $t$ on the path. The prefix discrepancy problem then corresponds to assigning signs to the vertices of $\mathsf{Path}_T$ and looking at the signed-sums of vectors associated with the vertices on any path starting from the root in $\mathsf{Path}_T$\footnote{The restriction of involving the root is only notational and can be relaxed at the cost of a factor of two.}. 

We investigate the case where the \emph{base graph} instead of being a simple path is an arbitrary DAG. 
In particular, given a DAG $G = ([T],E)$ with vertices $[T]:= \{1,\ldots, T\}$, we denote by $\prefix(G)$ the following set family on the ground set $[T]$: First order the vertices of $G$ topologically; now a set $S \subseteq [T]$ is included in $\prefix(G)$ iff $S$ is a path in $G$ that starts at the topologically ordered root vertex. For the generalization, given $T$ vectors in $\B_2^d$ corresponding to the $T$ vertices, our goal is to find a signing of the vertices (equivalently, the vectors) to minimize the maximum discrepancy with respect to any \emph{prefix constraint} in $\prefix(G)$. In other words, given input vectors $v_1,\ldots, v_T \in \B_2^d$,  find a signing $x \in \pmone^T$ that  minimizes 
\[ \pathdisc_G(v_1, \cdots, v_T) := \min_{x \in \pmone^T} \max_{S \in \prefix(G)}  \Big\|\sum_{t \in S} {x_t v_t} \Big\|_\infty. \]

Clearly, if $G$ is the simple path on $T$ vertices, we recover prefix discrepancy. 

We prove the following generalization of \Cref{thm:Banaszczyk12} in terms of the hereditary discrepancy, $\herdisc$, of the set system $\prefix(G)$, which is the largest possible discrepancy of any ``sub-DAG'' (see \Cref{def:herdisc}). This is a natural parameter since even for $d=1$ (i.e., scalars), $\herdisc$ is a lower bound on $\pathdisc_G$ if we take the vectors $\{v_t\}_{t \in [T]}$ to be $0/1$-valued scalars. 

\begin{restatable}[Prefix discrepancy for DAGs]{theorem}{dag}
\label{thm:DAGs_prefix}
Given a $T$-node DAG $G$ and input vectors $v_1, \ldots,v_T \in \B_2^d$, we have $
\pathdisc_G(v_1,\ldots,v_T) \leq \herdisc(\prefix(G)) \cdot O(\sqrt{\log d+ \log T})$. Additionally, if $d=1$, then $\pathdisc_G(v_1,\ldots,v_T) = O(\herdisc(\prefix(G)))$. 
\end{restatable}

\Cref{thm:DAGs_prefix} is a generalization of \Cref{thm:Banaszczyk12} because $\herdisc(\prefix(G)) = 1$ when $G$ is a path on $[T]$, as in Banaszczyk's setting. We prove the above theorem by showing a novel structural result about DAGs, which allows us to approximate the given DAG with rooted-trees with ``distance'' given by $\herdisc(\prefix(G))$ (see \Cref{lem:reduction_DAGs_to_trees}). Next, to handle any rooted-tree $\calT$,  we use techniques inspired from \cite{B12} and show that for any given vectors $v_1, \cdots, v_T \in \B_2^d$, we have  $\pathdisc_\calT(v_1,\cdots, v_T) = O(\sqrt{\log d + \log T})$ \footnote{It is easy to see that for any rooted tree $\calT$, we have that $\herdisc(\prefix(\calT))=1$, as in the case of prefixes of a path.}.

It is an intriguing open problem whether the dependence of $\sqrt{\log T}$ in Banaszcyzk's bound can be improved.  We show, however, that the $\sqrt{\log T}$  in \Cref{thm:DAGs_prefix} cannot be improved even for rooted-trees and even in the two-dimensional case. This gives an interesting contrast between trees and paths, even though they have the same 
$O(\sqrt{\log d + \log T})$ upper bound.

\begin{restatable}[Prefix Discrepancy Lower Bound]{theorem}{treeslb} \label{thm:prefix_lower_bound_trees}
For any integer $T > 0$, there exists a $T$-node rooted tree $\calT$ and input vectors $v_1,\ldots,v_T \in \B_2^2$ such that $\pathdisc_\calT(v_1,\cdots,v_T) = \Omega(\sqrt{\log T})$.  
\end{restatable}

Note that for $d=2$, \cite{BaranyGrinberg81} gives a bound of $O(1)$ on the prefix discrepancy problem (irrespective of $T$). Further, for the one-dimensional case, the discrepancy of any rooted tree is at most $1$.

One might wonder if our improvement to $\sqrt{\log d + \log\!\log T}$ in the smoothed setting (\Cref{thm:smoothed_upper_bound}) could also be generalized to DAGs. We show that such an improvement is impossible, in particular there is an $\Omega(\sqrt{\log T / \log\!\log T})$ lower bound even for trees in a completely stochastic setting. 

\begin{restatable}[Prefix Discrepancy Lower Bound in the Random Setting]{theorem}{treeslbrandom}
\label{thm:trees_random_lower_bound}
Let the dimension $d=2$. There exists a $T$-node rooted tree $\calT$ such that when unit vectors $v_t \in \mathbb{R}^d$ are drawn i.i.d. uniformly from $\S^{d-1}$, then w.h.p. every $\pm 1$-coloring satisfies  $\pathdisc_\calT(v_1, \cdots, v_T) = \Omega(\sqrt{\log T / \log\!\log T})$.
\end{restatable}







\subsection{Combinatorial Vector Balancing}

Finally, we generalize the vector balancing problem beyond prefixes and introduce the \emph{Combinatorial Vector Balancing} problem --- here the discrepancy constraints come from arbitrary set systems. In particular, given a set family $\CS$ on the  ground set $[T]$ and input vectors $v_1, \ldots, v_T \in \B_2^d$,  find a signing $x \in \pmone^T$ that minimizes
\[ \cdisc_\CS(v_1, \cdots, v_T) := \min_{x \in \{-1,+1\}^n} \max_{S \in \mathcal{S}} \Big\| \sum_{i \in S} x_i v_i \Big\|_\infty.\]

Notice that $\cdisc_\CS(v_1, \cdots, v_T)$ captures $\pathdisc_G(v_1, \cdots, v_T)$ when the set system $\mathcal{S}$ is given by $\prefix(G)$, and therefore also generalizes the prefix discrepancy problem. Moreover, $\cdisc_\CS(v_1, \cdots, v_T)$ can be viewed as a natural generalization of the combinatorial discrepancy of the set system $\CS$, where $v_1, \cdots, v_T$ are all scalars and equal to $1$. Combinatorial discrepancy of set systems has been well-studied in many different contexts and we refer the reader to the books \cite{Chazelle-Book01,Matousek-Book09} for an overview.

We also define the quantity
\begin{align*}
    \cdisc(\CS,d) := \sup_{v_1, \cdots, v_T \in \B_2^d} \cdisc_\CS(v_1, \cdots, v_T) \enspace ,
\end{align*}
which can  be viewed as a generalization of hereditary discrepancy of the set system $\CS$. 
In particular, for one dimensional input vectors, we have that $\herdisc(\CS) \le \cdisc(\CS, 1)$ by taking the corresponding scalars $\{v_t\}_{t \in [T]}$ to be $0/1$-valued. 

Our results from the previous section imply that if $\CS = \prefix(G)$ for some DAG $G$, then  $\cdisc(\CS,1)=\Theta(\herdisc(\CS))$ and $\cdisc(\CS,d) \le \herdisc(\CS)\cdot O(\sqrt{\log d + \log T})$. A natural question that comes up is whether an analogous relation also holds for a general set system $\CS$. We answer this question,  up to poly-logarithmic factors.



\begin{restatable}[Combinatorial Vector Balancing Upper Bound]{theorem}{combvecbal}
\label{thm:general-setsystem}
Let $\mathcal{S}$ be a set system on the ground set $[T]$. Given input vectors $v_1, \cdots, v_T \in \B_2^d$, we have that
\begin{align*}
   \cdisc_\CS(v_1, \cdots, v_T) \leq \herdisc(\mathcal{S})   \cdot \log  \big( \min\{ |\mathcal{S}|, T\} \big) \cdot \sqrt{\log d + \log |\mathcal{S}|} \enspace .
\end{align*} 
Moreover, the upper bound is constructive in the sense that there is a polynomial time algorithm to  find such a coloring. 
\end{restatable}

Note that the above also implies that $\cdisc(\mathcal{S},d)$ is bounded by the right hand side in the above theorem.

Our proof of Theorem~\ref{thm:general-setsystem} is via the $\gamma_2$-norm, which is used in~\cite{MNT14} as a tool to approximate $\herdisc(\mathcal{S})$. 
We show that since the vectors $v_1, \cdots, v_T$ are in  $\B_2^d$, the $\gamma_2$ norm of the left hand side (viewed as a $d |\mathcal{S}| \times T$ matrix) is at most $\gamma_2(A_{\mathcal{S}})$, where $A_{\mathcal{S}}$ is the incidence matrix of the set system $\mathcal{S}$. 
This allows us to  apply the connection between $\gamma_2$-norm and hereditary discrepancy in~\cite{MNT14} to prove Theorem~\ref{thm:general-setsystem}. More details on the proof can be found in Section~\ref{sec:comb_vec_bal}. 





\subsection{Technical Overview} \label{sec:overview}

\subsubsection{Smoothed Analysis of Prefix Discrepancy}

To achieve the $\sqrt{\log\!\log T}$ dependence in \Cref{thm:smoothed_upper_bound} in the smoothed setting, we group the input vectors into consecutive blocks of $n$ vectors, with a total of $T/n$ blocks, where $n = \poly(d, \log T)$.  Now
Banaszcyzk's result in Theorem~\ref{thm:Banaszczyk12}  implies that the prefix discrepancy within each block (viewed as an instance with $n$ input vectors) is $\Delta = O(\sqrt{\log d + \log n}) = O(\sqrt{\log d + \log\!\log T})$. However, the discrepancy for different blocks might add up arbitrarily. Our idea to get around this issue is to color the vectors block-by-block while using the current block to cancel any previously accumulated discrepancy.

\paragraph{Coloring Strategy for a Block.} We color the first block of $n$ vectors by using \Cref{thm:Banaszczyk12} and for the subsequent blocks, we adopt the following strategy. Suppose that we have assigned signs to the previous blocks and let $M \in \R^{d \times n}$ be the matrix whose columns are the current block of vectors that need to be colored. To find a coloring $x \in \pmone^n$ for a given block, we proceed in two phases: first we find a \emph{fractional} signing $x \in [-1,1]^n$ of the columns of $M$ so that 
\vspace{-\topsep}
\begin{enumerate} 
    \item We cancel out the discrepancy vector $w$ of the previous blocks, i.e., $Mx + w = 0$, and
    \item  The intermediate discrepancy $\| M_I x\|_\infty \leq O(\Delta)$ for every prefix interval $I \subseteq [n]$ and $M_I \in \R^{d \times n}$ is the matrix obtained from $M$ by zeroing out columns not in $I$.
\end{enumerate}

If we can find such a fractional solution, then a variant of the \emph{bit-by-bit rounding procedure} of~\cite{LSV86} allows us to round $x$ to a full coloring $x^* \in \{-1,1\}^n$ so that the prefix discrepancy for block $M$ due to the rounding is at most $\Delta:= \Theta(\sqrt{\log d + \log\!\log T})$  (Lemma~\ref{lem:rounding_in_block}).  
In particular, we show that if the previous discrepancy vector $w$ satisfies $\|w\|_\infty \leq \Delta$, then the discrepancy vector after block $M$ satisfies $\|w + Mx^*\|_\infty = \|M(x^* - x)\|_\infty \leq \Delta$, which allows us to proceed with coloring the next block inductively. The central part of the argument is to find  a fractional signing $x$ satisfying the above two properties for a given block, and we describe the ideas behind it next.

\paragraph{Primal-Dual Approach for Feasibility.} We need to show that for any vector $\|w\|_\infty \leq \Delta$ ($w$ is the discrepancy vector and the assumption on norm is the inductive hypothesis), with probability $1 - 1/\poly(T)$, we can find a fractional signing $x \in [-1,1]^n$ satisfying (a)~$Mx + w = 0$, and (b)~$\|M_I x\|_\infty \leq O(\Delta)$ for every prefix interval $I$.
These linear constraints are naturally captured by the feasibility of the following \emph{stochastic} linear program $\LP(\CI)$, where $\CI$ is the set of all prefixes of $[n]$ (see \Cref{subsubsec:feasibility}).
\begin{equation*}\tag*{$\LP(\CI)$}
    \begin{aligned} 
(Mx)_i & = -w_i &\quad \forall i\in [d]\\
 - 2\Delta  \leq (M_I x)_i & \leq 2\Delta &\quad \forall I \in \CI~,~ \forall i \in [d] \\
  -1 \leq x_j  & \leq 1 &\quad \forall j \in [n]
\end{aligned}
\end{equation*}

Our goal is to show that the above linear program is feasible with high probability over $M$. At a high-level this is reminiscent of the approach of \cite{BansalMeka-SODA19} who also use a stochastic LP to get cancellations, but the complexities are quite different: their LP is  simpler and they work in the \emph{completely random} setting as opposed to the more challenging smoothed-analysis setting as we do. Recall that for us, the columns of $M$ are generated in the smoothed analysis model: starting with an arbitrary matrix, we perturb the columns independently with random vectors whose covariance is lower bounded (in the PSD ordering) by $\epsilon^2 I_d$. Using a primal-dual approach, we derive the following sufficient condition for the feasibility of $\LP(\CI)$: For all dual variables $(y, \{\alpha_I\}_{I \in \CI})$ ($y, \alpha_I \in \R^d$) satisfying $\|y\|_1 = d$ and $\sum_{I\in \CI} \|\alpha_I\|_1 \le d/2$, we must have 
\begin{align}
    \label{eq:cond_intro} \|y^\top M - \sum_{I \in \mathcal{I}}  \alpha^\top_I M_{I}\|_1  ~\ge~ \Delta \cdot d  \enspace .
\end{align}
For intuition on why the above might be true, consider the case of a fixed choice of dual variables $(y, \{\alpha_I\}_{I \in \CI})$. 
Denoting $v_j$ as the $j$th column of $M$, the left hand side of \eqref{eq:cond_intro} can be written as $\sum_{j=1}^n  |(y - \sum_{I:j \in I}  \alpha_I)^\top v_j|$. The latter is a random variable (because of the perturbations in $v_j$'s) that is a sum of $n$ independent terms, each with variance at least $\epsilon^2 \cdot \|(y - \sum_{I:j \in I}  \alpha_I)\|_2^2$, while the right hand side of Equation \ref{eq:cond_intro} is only of order $\sqrt{\log n}$. 
Therefore, picking $n = \poly(d, \log T)$ to be large enough, Chernoff bound shows that the left hand side is larger than the right hand side with probability $1 - \exp(-n)$, for any fixed setting of the dual variables $(y, \{\alpha_I\}_{I \in \CI})$ that satisfies $\|y\|_1 = d$ and $\sum_{I\in \CI} \|\alpha_I\|_1 \le d/2$. 
However, there are $n$ prefix intervals in $\CI$ and each $\alpha_I$ needs to be discretized to multiples of $1/\poly(n,d)$, so the total number of choices of $(y, \{\alpha_I\}_{I \in \CI})$ is $\poly(n,d)^n \gg \exp(n)$. 
Thus, even though, inequality in Equation \ref{eq:cond_intro} holds with high probability for a fixed dual solution, we cannot use a union bound over all the choices of $(y, \{\alpha_I\}_{I \in \CI})$ to show that $\LP(\CI)$ is feasible with high probability.

\paragraph{Block Decomposition and Chaining.} To overcome the difficulty above, we work with a different family of intervals $I$ that we call a \emph{block decomposition} $\CIsub$. These intervals satisfy: 
\vspace{-\topsep}
\begin{enumerate}
    \item  Every prefix interval of $[n]$ can be written as a disjoint union of $I_1,I_2 \in \CIsub$ which allows us to deduce the feasibility of $\LP(\CI)$ from the feasibility of $\LP(\CIsub)$.
    \item There are only $o(n)$ ``long'' intervals in $\CIsub$. We are able to bypass the union bound by employing a two-step union bound, best thought of as a \emph{chaining} argument, to show that for any fixed choice of $(y, \{\alpha_I\}_{I\,\mathrm{long}})$, with probability at least $1 - \exp(-n)$, inequality \eqref{eq:cond_intro} holds for \textbf{all} choices of $\{\alpha_I\}_{I\,\mathrm{short}}$. We once again cannot do a union bound over all choices of $\{\alpha_I\}_{I\,\mathrm{short}}$, but use a more subtle argument based on the values of $\alpha_I$ to prove this claim (see Lemma \ref{lem:chaining}). Once we do so, we can take a naive union bound only over the $\exp(o(n))$ possible choices of $(y,\{\alpha_I\}_{I\,\mathrm{long}})$, thus showing that $\LP(\CIsub)$ is almost always feasible. The fractional signing $x \in [-1,1]^n$ can then be found by solving the corresponding LP feasibility problem.
\end{enumerate}



\subsubsection{Prefix Discrepancy for Trees and DAGs}

\paragraph{The upper bound.} 

The proof of \Cref{thm:DAGs_prefix} has two parts. First, we show that Banaszczyk's result in \Cref{thm:Banaszczyk12} can be generalized to prefixes of \emph{rooted trees} (which can be naturally viewed as DAGs by orienting the edges from the root to leaves). In particular we show that the same $O(\sqrt{\log d+ \log T})$ bound holds for prefixes of any rooted tree (\Cref{lem:prefix_trees}). This proof is inspired from Banaszczyk's original proof of \Cref{thm:Banaszczyk12}: we define intermediate auxiliary convex bodies inductively from the leaves to the root using a symmetrization operation for convex bodies defined in \cite{B12} (see \Cref{defn:ban_symmetrization}). We can show that all these intermediate bodies have large Gaussian measure and thus invoking the classic result of \cite{Banaszczyk-Journal98} gives us a coloring of the vectors with the above bound.

The second part is the crux of our proof, where we reduce the case of an arbitrary DAG $G$ to that of a rooted tree. This is where the notion of hereditary discrepancy of the set system $\prefix(G)$ comes in. In particular, we show:

\begin{restatable}[Reducing DAGs to Trees]{lemma}{DAGreduction}
\label{lem:reduction_DAGs_to_trees}
Given a DAG $G = ([T],E)$, there exists a (rooted) spanning tree  $\CT \subseteq G$ such that every directed path $\CP$ in $G$ has at most $4\cdot \herdisc(\prefix(G))$ non-tree edges. 
\end{restatable}

Since the removed edges break each path $\mathcal{P} \in \prefix(G)$ into at most $O\big(\herdisc(\prefix(G))\big)$ sub-paths in the tree $\CT$, the discrepancy along  $\mathcal{P}$ is  at most $\herdisc(\prefix(G)) \cdot O(\sqrt{\log d + \log T})$ using the bound for prefixes for trees (since any sub-path in the tree can be written as a difference of two prefixes). Below we describe the implications and ideas behind the above combinatorial characterization of hereditary discrepancy for prefixes of a DAG.



\paragraph{Hereditary Discrepancy for DAGs.} 

 \Cref{lem:reduction_DAGs_to_trees} also gives us a combinatorial characterization of $\herdisc(\prefix(G))$ as the proximity of DAG $G$ to a tree $\CT$, where the ``distance'' is the maximum number of non-tree edges along any path in $G$ (see \Cref{cor:comb_char_DAG}). This characterization might be of independent interest. 


For the proof of \Cref{lem:reduction_DAGs_to_trees}, we construct such a tree $\CT$ by going through the vertices of the DAG in the reverse topological order and adding edges inductively. However, it's not immediately clear what induction hypothesis we want to maintain throughout the process and how the number of non-tree edges can be related to $\herdisc(\prefix(G))$. The key observation, the one that brings $\herdisc(\prefix(G))$ into the picture is that $\herdisc(\prefix(G))$ is large if and only if $G$ contains a long {\em chain structure} (see \Cref{defn:chain_structure} for a formal definition). 
This is a subgraph $C$ that looks like the DAG shown in \Cref{fig:chain}, and it has the critical property that the set family $\prefix(C)$ is of exponential in $\ell$ size where $\ell$ is the length of the chain, implying that hereditary discrepancy of $\prefix(C) = \Omega(\ell)$.\\

\begin{figure}[h]
\centering 
\includegraphics[height=1in]{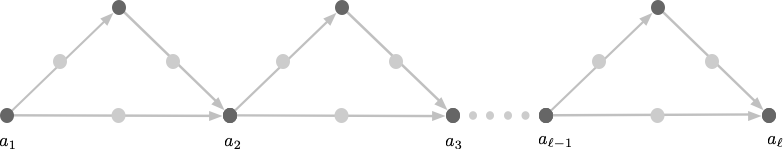}
\caption{\small An example of a chain structure of length $\ell-1$. Note that there are two vertex disjoint paths between $a_i$ and $a_{i+1}$ for each $i \in [\ell-1]$ and one of these contain at least one additional vertex apart from $a_i$ and $a_{i+1}$.}\label{fig:chain}
\end{figure}

 Using a careful induction, we are able to show that whenever during the construction of the tree $\CT$ in the reverse topological order, the maximum number of non-tree edges increases on some path, then the length of the longest chain and thus the hereditary discrepancy (of the prefixes of the sub-DAG) must also increase.

\paragraph{The lower bounds.} Both our lower bounds in \Cref{thm:prefix_lower_bound_trees,thm:trees_random_lower_bound} make use of a canonical construction. In particular, if one is considering prefixes of a complete binary tree with $T$ nodes, then there is a canonical construction of two-dimensional vectors $\{v_t\}_{t \in [T]}$ such that for any coloring $x \in \pmone^T$, there always exists a path so that the next vector on that path is always orthogonal to the current discrepancy vector.  Since the square of the $\ell_2$-norm of the discrepancy vector always increases as one goes down the path, this implies that the prefix discrepancy is at least $\Omega(\sqrt{h})$ where $h \approx \log T$ is the height of the tree. For our lower bound in the stochastic setting, we show that for any $\ell$-ary tree where $\ell=\polylog(T)$, if one chooses vectors $\{v_t\}_{t \in [T]}$ randomly, then with high probability, such a canonical instance is embedded in this random instance.



\newcommand{\Vol}{\mathsf{Vol}}

\section{Preliminaries} \label{sec:prelims}

\paragraph{Notation.}

Throughout this paper, we use $[k]$ to denote the set $\{1,2,\dotsc, k\}$ and $T \in \mathbb{Z}_{>0}$ to denote the number of vectors, which we assume to be $d$-dimensional. We use $\B_p^d$ to denote the unit $\ell_p$-ball in $\R^d$ and $\S^{d-1}$ to denote the unit Euclidean sphere in $\R^d$. We denote by $\gauss$ the $d$-dimensional standard Gaussian measure on $\BR^d$. For a random vector $v \in \R^d$, let $\mathsf{Cov}(v)$ denote the $d \times d$ covariance matrix of the distribution of $v$. For any two positive semidefinite  matrices $A$ and $B$, we use $A \succcurlyeq B$ to denote that $A - B$ is positive semidefinite. Given two DAGs $G_1$ and $G_2$ on the same vertex set, we write $G_1-G_2$ to denote the DAG which contains an edge $e$ iff $e \in G_1$ but $e\notin G_2$. We also write $G_2 \subseteq G_1$ if $G_2$ is a subgraph of $G_1$. 

\paragraph{Discrepancy.}

Let $\mathcal{F} =\{F_1, \cdots, F_m\}$ be a family of subsets over a ground set $[T]$, where $T \geq 1$ is a positive integer.
The discrepancy of $\mathcal{F}$ is defined as 
\begin{align*}
    \disc(\mathcal{F}) := \min_{x \in \{-1,1\}^T} \disc(\mathcal{F}, x),
\end{align*}
where $\disc(\mathcal{F}, x) = \max_{i \in [m]} |\sum_{j \in F_i} x_j|$. 
A vector $x \in \{-1,1\}^T$ is called a coloring (or signing) of $[T]$.

\begin{definition}[Hereditary discrepancy] \label{def:herdisc}
The hereditary discrepancy of $\mathcal{F}$ is defined as
\begin{align*}
    \herdisc(\mathcal{F}) := \max_{J \subseteq [T]} \disc(\mathcal{F}|_J) \enspace ,
\end{align*}
where $\mathcal{F}|_J := \{F_1 \cap J, \cdots, F_m \cap J\}$ is the restriction of the set system $\mathcal{F}$ to $J$.
\end{definition}

An important measure for approximating the hereditary discrepancy of a set system is the $\gamma_2$-norm (or the factorization norm), which is well-studied in mathematics. 

\begin{definition}[$\gamma_2$-norm]
Let $A \in \R^{m \times n}$ be a matrix. The $\gamma_2$-norm of $A$ is defined as 
\begin{align*}
    \gamma_2(A) := \min \big\{r(B) \cdot c(C) \mid A = BC \big\}  \enspace ,
\end{align*}
where $r(B)$ is the largest $\ell_2$-norm of the rows of $B$, and $c(C)$ is the largest $\ell_2$-norm of the columns of $C$, and the minimum is taken over all real matrices $B$ and $C$ such that $A = BC$. 
\end{definition}

The following relationship between $\gamma_2$-norm and $\herdisc$ was proved in~\cite{MNT14}.

\begin{lemma}[\cite{MNT14}]
\label{lemm:gamma2_herdisc}
Let $A \in \mathbb{R}^{m \times T}$ with rank $r$. Then, we have $\frac{\gamma_2(A)}{\log r} \leq \herdisc(A) \leq \sqrt{\log m} \cdot \gamma_2(A)$. 
\end{lemma}

\paragraph{Banaszczyk's Symmetrization.}

The following definition from~\cite{Banaszczyk-Journal98} is crucially used to obtain the bound in \Cref{thm:Banaszczyk12}.
Let $K \subseteq \mathbb{R}^d$ be a convex body and $u \in \mathbb{R}^d$ be a non-zero vector. 
Let $Z$ be the cylinder (possibly empty)  consisting of all lines $\ell$ parallel to $u$ such that the Euclidean length of the segment $K \cap \ell$ is greater than or equal to $2 \|u\|_2$. 
If $u = 0$, then $Z$ is defined to be equal to $K$.

\begin{definition}[Banaszczyk's symmetrization] \label{defn:ban_symmetrization}
Let $K \subseteq \R^d$ be a convex body, $u \in \R^d$ a vector, and $Z$ as defined above. We define $K * u$ as
\begin{align*}
    K * u := \big[(K + u) \cup (K - u) \big] \cap Z  \enspace .
\end{align*}
\end{definition}
In particular, it follows that $K*u \subseteq (K+u) \cup (K-u)$ and that if $K$ is symmetric , then so is $K * u$.
The following lemma is central to the proof of \Cref{thm:Banaszczyk12}.

\begin{lemma}[\cite{Banaszczyk-Journal98}]\label{lem:ban_symmetrization}
Let $K \in \R^d$ be a symmetric convex body with Gaussian measure $\gauss(K) \geq \frac12$, and let $u \in \frac{1}{5} \cdot \B_2^d$, then we have $\gauss(K * u) \geq \gauss(K)$. 
\end{lemma}

\paragraph{Basic Probabilistic Lemmas.} Here we collect some basic probabilistic facts that will be used later.

\begin{proposition}\label{prop:cap}
Let $d \ge 2$, and $\eps \le 1/d$, and $X\in \R^d$ be sampled uniformly from the unit sphere $\S^{d-1}$. Then, for any fixed vector $u \in \S^{d-1}$, we have
  \[  \p[\|X - u\|_2 \leq \eps] ~\geq ~ \exp(-10 d \log (1/\eps))  \enspace . \]
\end{proposition}
\begin{proof}
   Since $\|u\|_2=1$, the event $\|X - u\|_2 \leq \eps$ is equivalent to saying that $X^\top u \ge 1-\eps$. Thus, the probability of this event is exactly equal to the ratio  $\Vol(\CS)/\Vol(\S^{d-1})$ where $\Vol$ denotes the surface area and $\CS = \{ x \in \S^{d-1} \mid x^\top u \ge 1-\eps\}$ is the spherical cap. Using Lemma 4.1 in \cite{MV10}, we get that this ratio is at least $\eps^{10d} = \exp(-10 d \log (1/\eps))$.  
\end{proof}

\begin{proposition}\label{prop:rvec}
   Let $X \in 2\cdot \B_2^d$ be a random vector such that $\Cov(X)\succcurlyeq \eps^2 I$, then for any $u \in \R^d$,  
    \[ \BP\left[ |u^\top X| \ge \frac 12 \cdot \eps \|u\|_2\right] ~\ge~ \frac{\eps^2}{10} \enspace .\]
\end{proposition}
\begin{proof}
    We may assume that $\|u\|_2=1$. Let us denote the mean and covariance of $X$ by $\bmu$ and $\bCov$. Note that 
    \[ \BE[|u^\top X|^2] ~=~ u^\top \BE[XX^\top] u ~=~ u^\top \bCov u + |u^\top \bmu|^2 ~\ge~ \eps^2  \enspace ,\]
     since $\bCov = \BE[XX^\top]-\bmu\bmu^\top$. Since $\|X\|_2 \le 2$ by assumption, $|u^\top X|\le 2$ always, and a straightforward averaging argument yields the statement of the proposition. 
\end{proof}

\section{Smoothed Analysis of Prefix Discrepancy}

In \Cref{subsec:smoothed_proof}  by showing an upper bound of $O(\sqrt{\log d + \log\!\log T})$ for prefix discrepancy in the smoothed setting. Recall  that in this setting the input vectors  $v_t= \overline{v}_t + \hat{v}_t$, where $\overline{v}_t$ are worst-case vectors and $\hat{v}_t$ are small random perturbations.


\smooth*

Next in \Cref{sec:smoothlb}, we show that the bound in the above theorem  cannot be improved even for the stochastic setting (which is a special case of smoothed setting), without further progress on Banaszczyk's result on prefix discrepancy for adversarial vectors (\Cref{thm:Banaszczyk12}). In particular, we show that if unit vectors $v_1, \cdots, v_T \in \S^{d-1}$ are drawn uniformly, then the bound from \Cref{thm:smoothed_upper_bound} is attained with high probability assuming the tightness of Banaszczyk's bound.



\begin{thm}[Conditional Lower Bound for Smoothed/Stochastic Setting]
\label{thm:stochastic_lower_bound}
Let $d \ge 2$ and $n\ge \log d$ be integers. Suppose that there is an adversarial configuration of unit vectors $w_1, \ldots, w_n \in \S^{d-1}$ satisfying $\predisc(w_1, \cdots, w_t) = \Omega(\sqrt{\log d + \log n})$.
Let $T = n \cdot \exp(100nd \log n)$ and let unit vectors $v_1, \ldots, v_T \in \S^{d-1}$ be drawn uniformly, then with probability at least $1  - e^{-\sqrt{T}}$, we have that $\predisc(v_1, \cdots, v_T) = \Omega(\sqrt{\log d + \log\!\log T})$.
 \end{thm}



\subsection{The Upper Bound}
\label{subsec:smoothed_proof}

To prove \Cref{thm:smoothed_upper_bound}, 
let us start by grouping the input vectors $v_1, \cdots, v_T$ into consecutive blocks of $n$ vectors, with a total of $T/n$ blocks (without loss of generality, assume here that $T/n$ is an integer).
We will set $n=\poly(d,\log T)$ with the exact parameters to be specified later. Let  $\Delta = \Delta(d,n) = \Theta(\sqrt{\log d + \log n})$ be twice the prefix discrepancy upper bound given by Banaszcyzk's Theorem~\ref{thm:Banaszczyk12} for coloring $n$ vectors.
Our strategy will work in rounds and in round $r \in [T/n]$ it will give signs to the vectors in the $r^{\text{th}}$ block. 
We maintain the following invariants at the end of round $r$:
\begin{description}
    \item[\textbf{Invariant 1}:] Discrepancy at the end of $r^{\text{th}}$ block is $\infnorm{\sum_{ t\le rn} x_t v_t} \le \Delta$
    \item[\textbf{Invariant 2}:] Discrepancy of any prefix of the first $r$ blocks is $\max_{\tau \in [rn]} \infnorm{\sum_{ t \le \tau} x_t v_t} \le 6\Delta$.
\end{description}
 The first condition in the invariant allows us to inductively proceed with coloring the next block of vectors while the second condition gives us the desired bound on the maximum discrepancy of any prefix.


The first round of our strategy gives signs to the first block of $n$ vectors using the result of \cite{B12}, which maintains both invariants. Each other round $1<r<T/n$ will have two phases: in the first phase we find a fractional signing of the next block of vectors and then round it to obtain a $\pmone$ signing. 

To describe the two phases in round $r$, let us use $M \in \R^{d \times n}$ to denote the matrix whose $n$ columns are the vectors being processed in round $r$, \emph{i.e.} $v_{(r-1)n+1}, \ldots, v_{rn}$. For any subset $I \subseteq [n]$, let us also write $M_I$ to denote the $d \times n$ matrix where we keep the $j^{\text{th}}$ column of $M$ if it occurs in $I$ and replace it with zeros otherwise. We also define $w := \sum_{t \le (r-1) n} x_t v_t$ to be the discrepancy vector at the end of round $r-1$. 

In the first phase, we find a fractional signing satisfying some properties that help in maintaining the invariants after round $r$.


%


\begin{lemma}[Fractional signing] \label{lem:block_induction}
For an arbitrary $\infnorm{w} \leq \Delta$,  with probability at least $1 - 1/\poly(T)$ there is a fractional signing $x \in [-1,1]^n$ 
satisfying:
\vspace{-\topsep}
\begin{enumerate}
    \item $Mx + w = 0$, and
    \item $\|M_I x\|_\infty \leq 4\Delta$ for every prefix interval $I \subseteq [n]$, \emph{i.e.}, $I \in \{[k] \mid k \le n\}$.
\end{enumerate}
\end{lemma}

The conditions above are linear constraints, so the lemma above shows that the corresponding linear program is feasible with high probability. The natural approach to showing this using strong duality encounters some obstacles (elaborated on later), so our proof relies on a stronger linear program where one imposes the second condition in the lemma for a different family of intervals. Using this stronger linear program and a careful chaining argument, we can  infer the existence of a feasible solution satisfying the conditions above with high probability.

The second phase rounds the fractional signing to a $\pm 1$ signing using the following lemma, which is similar to the bit-by-bit rounding procedure used by Lov{\'a}sz, Spencer, and Vesztergombi to establish the connection between hereditary and linear discrepancy \cite{LSV86}.


\begin{lemma}[Block Rounding] \label{lem:rounding_in_block}
Let $x \in [-1,1]^n$ be a fractional signing. There exists a signing $x^* \in \{-1,1\}^n$ such that $\infnorm{M_{I} (x^* - x)} \leq \Delta$ for any prefix interval $I \subseteq [n]$.
\end{lemma}

We will prove \Cref{lem:block_induction} and \Cref{lem:rounding_in_block} in \Cref{subsubsec:feasibility} and \Cref{sec:rounding}, but first we show how they imply \Cref{thm:smoothed_upper_bound}.


\begin{proof}[Proof of \Cref{thm:smoothed_upper_bound}]
Assuming that the two invariants are maintained at the end of round $r-1$, we show that they hold after round $r$ as well. For the first invariant, the discrepancy vector at the end of round $r$ is $w + Mx^*$. Therefore, using that $\infnorm{w} \le \Delta$, the first condition of \Cref{lem:block_induction} and \Cref{lem:rounding_in_block}, we have
\begin{align*}
\left\|w + Mx^* \right\|_\infty ~=~  \|w + Mx + M(x^* - x)\|_\infty ~=~ \| M(x^* - x)\|_\infty ~\leq~ \Delta  \enspace .
\end{align*}



For the second invariant, for each prefix of the first $(r-1)$ blocks the discrepancy remains the same, so let us consider a prefix of the form $w+M_I x^*$ where $I$ is some prefix interval of $[n]$. In this case, using the second condition of \Cref{lem:block_induction}, we can bound
\begin{align*}
\| w + M_{I} x^* \|_\infty ~\leq~ \|w\|_\infty + \| M_{I} (x^* - x) \|_\infty + \|M_{I} x\|_\infty ~\leq~ 6 \Delta  \enspace . \altqedhere
\end{align*}
\end{proof}

\subsubsection{Proof of \Cref{lem:rounding_in_block}} \label{sec:rounding}

The proof of Lemma~\ref{lem:rounding_in_block} is a slight variant of the proof used in~\cite{LSV86} to upper bound linear discrepancy by hereditary discrepancy. 
We give it here for completeness. 

\begin{proof}[Proof of Lemma~\ref{lem:rounding_in_block}]
We first apply the coordinate-wise map $y = (x + 1)/2 \in [0,1]^n$ for the fractional signing $x \in [-1,1]^n$. 
Our goal is to show that there exists $y^* \in \{0,1\}^n$ such that for any prefix interval $I \subseteq [n]$, we have $\|M_I (y^* - y)\| \leq \Delta/2$. 
The lemma immediately follows by taking the inverse map $x^* = 2y^* - 1 \in \{-1,1\}^n$.

Write the binary expansion of each $y_i$ as $y_i = \sum_{j \geq 1} q_{i,j} 2^{-j}$, where $q_{i,j} \in \{0,1\}$ denotes the $j$-th bit of $y_i$ after the decimal point. 
Let $k$ be a positive integer, and assume for now that all the $y_i$'s have at most $k$ bits in their binary expansions (if the binary expansion of some $y_i$ has $k_i < k$ bits, we let $q_{i,k_i + 1} = \cdots = q_{i,k} = 0$).  
We shall eventually send $k \rightarrow \infty$ so the above assumption becomes without loss of generality. 
Consider the $k$th bit $q_{i,k}$ of each $y_i$. Let $M^{(k)}$ be
the sub-matrix of $M$ restricted to those columns $i$ for which $q_{i,k} = 1$. By applying Banaszczyk's result in \Cref{thm:Banaszczyk12} to the columns of the matrix $M^{(k)}$, there is a $\pm 1$ signing $\chi^{(k)}$ of the columns of $M^{(k)}$ such that the prefix discrepancy of $M^{(k)}$ is at most $\Delta/2$. Viewing $\chi^{(k)}$
as a vector in $\R^n$ with entries $\{-1,0,1\}$ (where the 0 entries correspond to the columns
not in $M^{(k)}$, we update the vector $y' = y +2^{-k} \chi^{(k)}$.
Now note that the $k$-th bit of each $y_i'$ is $0$, as $\chi^{(k)}_i \in \{-1,1\}$ if $q_{i,k} = 1$ and $0$ if $q_{i,k} = 0$. 
Moreover, for any prefix interval $I \subseteq [n]$, we have $\| M_I (y' - y)\|_\infty = 2^{-k} \cdot \| M_I \chi^{(k)} \| \leq 2^{-k} \cdot \Delta/2$.

We can now iterate this bit-rounding procedure to the bits $k-1, k-2, \cdots, 1$ (by treating $y'$ as the new vector $y$ with only $k-1$ bits).
This produces a $\{0,1\}$-vector $y^* = y + 2^{-k} \chi^{(k)} + 2^{-k+1} \chi^{(k-1)} + \cdots + 2^{-1} \chi^{(1)}$ such that
$\| M_I (y^* - y) \|_\infty \leq (2^{-k} + 2^{-k + 1} + \cdots + 2^{-1}) \cdot \Delta/2 \leq \Delta/2$ holds for any prefix interval $I \subseteq [n]$. 
As this rounding error does not depend on $k$, sending $k \rightarrow \infty$ completes the proof of the lemma. 
\end{proof}

\subsubsection{Proof of \Cref{lem:block_induction} via a Primal-dual Approach} \label{subsubsec:feasibility}

As mentioned before, for our analysis it will be useful to consider a stronger linear program which imposes constraints on a different family of intervals other than just the prefix intervals of $[n]$. With this in mind, for a family $\CI$ of intervals of $[n]$, we consider the following linear programming feasibility problem:
\begin{equation*}\tag*{$\LP(\CI)$}
    \begin{aligned} 
\text{max} \quad &\quad 0  &    \\
\text{s.t.} \quad  &(Mx)_i  = -w_i &\quad \forall i\in [d]\\
 &-2\Delta \leq (M_I x)_i  \leq 2\Delta &\quad \forall I \in \CI, i \in [d] \\
  &-1\leq x_j   \leq 1 &\quad \forall j \in [n]
\end{aligned}
\end{equation*}


To understand when $\LP(\CI)$ is feasible, we consider the corresponding dual linear program:
%
\begin{equation*}
    \begin{aligned}
    \text{min} \quad -y^\top w + 2\Delta \sum_{I \in \mathcal{I}}  (\alpha^+_I + \alpha^-_I)^{\top}\bone + (\gamma^{+} + \gamma^{-})^\top \bone  & \\
    \text{s.t.} \qquad\quad y^\top M + \sum_{I \in \mathcal{I}} (\alpha_I^{+}-\alpha_I^{-})^{\top} M_I + (\gamma^{+}-\gamma^{-})^\top &=\bzero \\
\alpha_I^+, \alpha_I^-, \gamma^+, \gamma^- & \geq \bzero  \qquad \forall I \in \mathcal{I} .
    \ 
\end{aligned}
\end{equation*}


By strong duality, the primal is infeasible iff the dual objective is strictly negative. Next, we simplify the dual program and derive an analytical condition which will imply the feasibility of $\LP(\CI)$. For this, we note that for any fixed value of the vector $\alpha_I := \alpha_I^{+}-\alpha_I^{-}$, with $\alpha_I^+$ and $\alpha_I^-$ non-negative, the least value of 
$(\alpha_I^{+}  + \alpha_I^{-})^\top \bone$ is $\|\alpha_I\|_1$. This holds similarly for $\gamma := \gamma^+ - \gamma^-$. So, the dual objective can be simplified as 
\[-y^\top w + 2\Delta \sum_{I \in \mathcal{I}} \|\alpha_I\|_1 + \|\gamma\|_1  \enspace .\]
By the dual constraint $\|\gamma\|_1 = \|y^\top M + \sum_{I \in \mathcal{I}} \alpha_I^\top M_I \|_1$, the primal is \emph{infeasible} if and only if there exist $y,\alpha_I \in \R^d$ for $I \in \mathcal{I}$ such that 
\[-y^\top w + 2\Delta \sum_{I \in \mathcal{I}} \|\alpha_I\|_1  + \|y^\top M + \sum_{I \in \mathcal{I}}  \alpha_I^\top M_I\|_1 ~<~ 0  \enspace .\]

Replacing $\alpha_I$ by $-\alpha_I$ and rearranging the above, we get that $\LP(\CI)$ is feasible if  for every $y, \{\alpha_I\}_{I \in \CI} \in \R^d$ the following holds
  \begin{equation}\label{eq:cond1}
    \|y^\top M - \sum_{I \in \mathcal{I}}  \alpha^\top_I M_{I}\|_1  ~\geq~ y^\top w  - 2\Delta\sum_{I \in \mathcal{I}} \|\alpha_I\|_1  \enspace .
\end{equation}

Using our assumption that $\infnorm{w} \le \Delta$, the above condition is satisfied if the left hand side in \eqref{eq:cond1} is at least $\Delta(\|y\|_1  - 2\sum_{I \in \mathcal{I}} \|\alpha_I\|_1)$. We may further restrict our attention to the case when $\|y\|_1  \geq 2\sum_{I \in \mathcal{I}} \|\alpha_I\|_1$ since otherwise the inequality holds trivially. Furthermore, by rescaling, we may assume that $\|y\|_1 = d$. Thus, denoting by $K \subseteq \R^{d + |\CI|d}$ the set of all vectors $(y, \{\alpha_I\}_{I \in \CI})$ satisfying $\|y\|_1 = d$ and $\sum_{I\in \CI} \|\alpha_I\|_1 \le d/2$, we arrive at the following sufficient condition for feasibility.

\begin{proposition}\label{prop:cond}
   If $\infnorm{w} \le \Delta$ and for every $(y, \{\alpha_I\}_{I \in \CI})  \in K$ we have
    \begin{equation}
    \label{eq:cond} \|y^\top M - \sum_{I \in \mathcal{I}}  \alpha^\top_I M_{I}\|_1  ~\geq~ \Delta \cdot d \enspace ,
\end{equation}
then \eqref{eq:cond1} holds as well, and thus $\LP(\CI)$ is feasible.
\end{proposition}

Note that the left hand side of \eqref{eq:cond} is 
\begin{equation}\label{eqn:prefixsum}
    \ \sum_{j=1}^n  |(y - \sum_{I:j \in I}  \alpha_I)^\top v_j| ~=~ \sum_{j=1}^n |z_j^\top v_j| \enspace ,
\end{equation}
defining $z_j := y - \sum_{I: j\in I} \alpha_I$ where $z_j \in \R^d$ satisfy $\|z_j\|_1 \ge \|y\|_1 - \sum_{i \in \CI} \|\alpha_I\|_1 \ge d/2$ for every $j \in [n]$. If there is sufficient randomness in the vectors $v_j$, so that each term $|z_j^\top v_j|=1/\poly(d,\log T)$ with constant probability, by taking $n = \poly(d, \log T)$, we can expect that \eqref{eq:cond} holds (recall that $\Delta = O(\sqrt{\log(dn)})$ with probability $1-e^{-cn}$  for any given $(y, \{\alpha_I\}_{I \in \CI}) \in K$. 

Given this, it is natural to attempt to take $\CI$ to be $\CI_{\mathsf{prefix}}$, the set of all prefixes of $[n]$, as required in our target \Cref{lem:block_induction}. However, naively this requires us to perform a union bound over a set of vectors that is much larger than the $e^{-cn}$ probability of the bad event. In particular, we need to do  a union bound over all possible choices  $(y, \{\alpha_I\}_{I \in \CI}) \in K$. Since it suffices to discretize these vectors to have values which are integer multiples of $1/\poly(n,d)$, given that $|\CI_{\mathsf{prefix}}|=n$, the total number of choices is $\poly(n,d)^{nd} \gg e^{cn}$. It is unclear how to make this strategy work, so we work with a different family of intervals $\CIsub$ where the constraints not only allow us to infer the target constraints of \Cref{lem:block_induction} but a careful chaining argument allows a more refined union bound.

\paragraph{Improved Bounds using a Block Decomposition $\mathcal{I}_{\textsf{block}}$.} \label{subsubsec:sub-block_decomposition}
 
Let us partition $[n]$ into consecutive blocks of length $b=n^{0.1}$ (for convenience, we assume that $b$ divides $n$). Consider the following family $\CIsub$ of intervals of $[n]$ obtained by taking 
\begin{itemize}
    \item all prefix intervals of $[n]$ which end after an entire block, \emph{i.e.}, intervals of the form  $[ib]$ where $i \in [n/b]$, and
    \item all prefix intervals within each block, \emph{i.e.}, intervals of the form $\{(i-1)b+1,\ldots,(i-1)b+r\}$ where $i \in [b]$ and $r\in[n/b]$.
\end{itemize}

Note that any prefix interval $I \in \CI_{\textsf{prefix}}$ can be expressed as the union of two intervals $I_1, I_2  \in \CIsub$. Thus, if $\LP(\CIsub)$ is feasible we get that $Mx+w=0$ and also $\infnorm{M_Ix} \le \infnorm{M_{I_1}x} + \infnorm{M_{I_2}x} \le 4\Delta$, thus implying the constraints of \Cref{lem:block_induction}. 

To complete the proof, we show that $\LP(\CIsub)$ is feasible w.h.p. by bounding the left hand side of \eqref{eq:cond}. For this let us introduce some terminology: we say an interval $I \in \CIsub$ is $\emph{long}$ if $|I|\ge b$ and \emph{short} otherwise. Note that the size of the family $|\CIsub|=\Theta(n)$ but the number of long intervals is $2n/b = o(n)$ and furthermore, any coordinate $j \in [n]$ is only contained in $n/b$ short intervals. 

To bound \eqref{eq:cond}, we use a two-step union bound, best viewed as a chaining argument. In particular, we show the following lemma where we only consider $(y, \{\alpha_I\}) \in K$ as in \Cref{prop:cond}.

\begin{claim}\label{lem:chaining}
    Fix any $(y, \{\alpha_I\}_{I\,\mathrm{long}})$. Then, for $n=\poly(d, \log T)$ sufficiently large, the probability that \eqref{eq:cond} holds for \textbf{all} choices of $\{\alpha_I\}_{I\,\mathrm{short}}$ is at least $1-e^{-c \epsilon^2 n}$ for a universal constant $c>0$.
\end{claim}

Note that there are $\Theta(n)$ short intervals and  performing a naive union bound over all possible choices of $\{\alpha_I\}_{I \,\mathrm{short}}$ (discretized as before) is not enough to prove the above claim and our proof below will use a more refined argument. However, given the above claim, to show that \eqref{eq:cond} holds for all $(y, \{\alpha_I\}_{I \in \CIsub}) \in K$, we only need to do a union bound over all the possible choices of $(y, \{\alpha_I\}_{I \,\mathrm{long}})$ where the number of long intervals is $O(n/b)=O(n^{0.9})$.

More precisely, discretizing each vector to integer multiples of $1/\poly(n,d)$, the number of possible choices of $(y, \{\alpha_I\}_{I \,\mathrm{long}})$ is $(nd)^{O(n/b)}=e^{O(n^{0.91})}$ as $d \le n$. Since $\epsilon = 1/\poly(d, \log T)$, choosing $n=\poly(d,\log T)$ to be a large enough polynomial, we get that $\LP(\CIsub)$ is feasible with probability at least $1- e^{-\Omega(n^{0.05})} = 1-T^{-2}$.

To finish the proof, we now  prove \Cref{lem:chaining}.

\begin{proof}[Proof of \Cref{lem:chaining}]

For $j \in [n]$, let us define $\zl{j}:= y - \sum_{I: j \in I, I \,\mathrm{long} } \alpha_I$ and $\zs{j}:= \sum_{I: j \in I, I \,\mathrm{short} } \alpha_I$. For a fixed $(y, \{\alpha_I\}_{I \in \CI})$, the contribution of coordinate $j \in [n]$ to the LHS of \eqref{eq:cond} is  
\begin{align}\label{eqn:light}
    \   \Big| \Big(y-\sum_{I: j \in I} \alpha_I \Big)^\top v_j \Big| &~=~ \left|(\zl{j} - \zs{j})^\top v_j \right| ~\ge~ \min\left\{0, \left|\zl{j}^{\top}{v_j}\right| - \left|\zs{j}^{\top}{v_j}\right|\right\}  \enspace .
\end{align}

First we show that the contribution of $|\zl{j}^\top v_j|$ is typically large for most $j$. For this, we note that $\|\zl{j} \|_1 \geq \|y\|_1 - d/2 \ge d/2$, and so $\|\zl{j}\|_2 \geq \sqrt{d}/2$. By our assumption on the noise, we have $\|v_j\|_2 \leq 2$  and that $\Cov(v_j)\ge \eps^2 I_d$. Thus, \Cref{prop:rvec} implies that for any given $j \in [n]$,
\begin{align} \label{eq:large_coordinate}
    \p\Big[|\zl{j}^\top v_j| \geq \delta \Big] \geq \frac{\eps^2}{10}  \enspace ,
\end{align}
where $\delta = \frac{1}{4} \cdot \eps \sqrt{d}$.

To bound the contribution of $|\zs{j}^\top v_j|$, we define some terminology. Given a setting of $(y, \{\alpha_I\}_{I\in \CI})$, we say that a coordinate $j \in [n]$ is \emph{heavy} if it lies in at least one \emph{short} interval $I$ such that $\|\alpha_I\|_1 \ge \frac{\eps}{100b}$. We say that $j$ is {\em light} if it is not heavy. Since $\sum_I \|\alpha_I\|_1 \leq d/2$, there are at most $\frac{50bd}{\eps}$ intervals $I$ with $\|\alpha_I\|_1 \ge \frac{\eps}{100b}$. 
As every short interval contains at most $b$ coordinates, 
it follows that at most $\frac{50b^2d}{\eps}$ coordinates are heavy. Since $b=n^{0.1}$ and $\eps=1/\poly(d,\log T)$, choosing $n=\poly(d, 1/\eps) = \poly(d,\log T)$ to be big enough, we can guarantee that at most $50bd/\eps \leq \epsilon^2 n/ 100$ coordinates $j \in [n]$ are heavy.




If $j$ is light, then all short intervals $I$ that contain $j$ satisfy $\|\alpha_I\|_1 \le \eps/(100b)$. 
Note that there are at most $b$ short intervals that contain $j$. Thus, $\|\zs{j}\|_1 \le \eps/100$ and so, using Cauchy-Schwarz, we have 
\begin{align}\label{eqn:smallval}
|\zs{j}^\top v_j| & ~\le~ \|\zs{j}\|_2 \cdot \|v_j\|_2 ~\le~ \sqrt{d}~ \|\zs{j}\|_1 \cdot 2\le \frac{\eps \sqrt{d}}{50} ~\le~ \frac{\delta}{10}  \enspace .
\end{align}

Now comes the simple but important observation: the value $\zl{j}$ only depends on $(y, \{\alpha_I\}_{I \,\mathrm{long}})$, but not on any $\alpha_I$ where $I$ is a short interval. 
In particular, defining the event (which only depends on $(y, \{\alpha_I\}_{I \,\mathrm{long}})$)
\[\CE ~:=~ \CE(y, \{a_I\}_{I \,\mathrm{long}}) ~=~ \big\{\text{at least } \epsilon^2 n /20 \text{ coordinates } j \in [n] \text{ satisfy } |\zl{j}^\top v_j| \geq \delta \big\}  \enspace ,\] 
we know from \eqref{eq:large_coordinate} and the multiplicative form of Chernoff bound that $\p[ \CE ] \geq 1 - e^{-c \epsilon^2 n}$ for a universal constant $c>0$.

When the event $\CE$ is true, we claim that \eqref{eq:cond} always holds. In particular, under the event $\CE$, there are at least $\epsilon^2 n/20$ coordinates $j \in [n]$ with $|\zl{j}^\top v_j| \geq \delta$. Among these, at most $\epsilon^2 n/100$ coordinates $j$ are heavy, so overall at least $\epsilon^2 n/25$ coordinates satisfy $|\zl{j}^\top v_j| \geq \delta$ and are light. Using \eqref{eqn:light} and \eqref{eqn:smallval}, any such coordinate contributes at least $\delta/2$ to the left hand side of \eqref{eqn:light}. It follows immediately that the left hand side of \eqref{eq:cond} is $\Omega(\epsilon^2 n\delta) = \Omega(n\cdot \eps^3 \sqrt{d})$. Since $\epsilon = 1/\poly(d, \log T)$, taking $n=\poly(d,1/\eps, \log T)=\poly(d, \log T)$ suffices to ensure that this is much bigger than $\Delta d$. This completes the proof of \Cref{lem:chaining}. \qedhere

\end{proof}

\subsection{Proof of the lower bound}
\label{sec:smoothlb}

\begin{proof}[Proof of \Cref{thm:stochastic_lower_bound}]
The main idea is to partition the $T$ vectors into $T/n$ consecutive blocks $M_i$ for $i\in [T/n]$ and show that with probability at least $1 -e^{-\sqrt{T}}$, one of the blocks is very close to the adversarial instance $w_1, \cdots, w_n \in \R^d$ whose existence we assumed.

Since $\|w_t\|_2=1$ for each $t$, using \Cref{prop:cap}, we have that for $v_1, \cdots, v_n$ i.i.d. sampled uniformly at random from $\mathbb{S}^{d-1}$, the following holds
\begin{align*}
    \p \big[\|v_t - w_t\|_2 \leq 1/n, \forall t \in [n] \big] ~\geq~ \exp(- 10 n d \log n)  \enspace .
\end{align*}
Note that when $\| v_t - w_t\|_2 \leq 1/n$, the $d \times n$ matrix formed by the vectors $v_1, \cdots, v_n$ also has prefix discrepancy $\Omega(\sqrt{\log d+ \log n})$. Thus, each block has prefix discrepancy $\Omega(\sqrt{\log d+ \log n})$ with probability at least $\exp(-10 nd \log n)$. Since we have $T/n = \exp(100nd\log n)$ blocks, the probability that none of the blocks has prefix discrepancy $\Omega(\sqrt{\log d+ \log n})$ is upper bounded by
\begin{align*}
    \left(1 - \exp(-10 nd \log n) \right)^{T/n}~ =~ \left(1-\left(n/T\right)^{0.1}\right)^{T/n} ~\leq~ e^{-(T/n)^{0.9}} ~\le~ e^{-\sqrt{T}} \enspace .
\end{align*}
Whenever there exists a block $M_i$ with prefix discrepancy at least $\Omega(\sqrt{\log d+\log n})$, then the prefix discrepancy for all the input vectors in blocks $M_1, \cdots, M_i$ is also $\Omega(\sqrt{\log d+ \log n})$ as each prefix in $M_i$ can be written as the difference between two prefixes of $M_1, \cdots, M_i$. This finishes the proof of the theorem. 
\end{proof}




\section{Prefix Discrepancy for Trees and DAGs}




In this section, we prove our results about prefix discrepancy for DAGs. Our first result, whose proof is given in \Cref{sec:dagup}, is \Cref{thm:DAGs_prefix} which gives an upper bound on prefix discrepancy for DAGs.

\dag*

Next in \Cref{sec:dag-lb}, we prove our lower bounds for prefix discrepancy of DAGs. We first show that the $\sqrt{\log T}$  in \Cref{thm:DAGs_prefix} cannot be improved in general, even for rooted-trees.

\treeslb*

One might wonder if our improvement to $\sqrt{\log d + \log\!\log T}$ in the smoothed setting (\Cref{thm:smoothed_upper_bound}) could also be generalized to DAGs. 
We show that such an improvement is impossible, even for trees in the completely stochastic setting. 

\treeslbrandom*

The proof of  \Cref{thm:prefix_lower_bound_trees} is given in \Cref{subsec:prefix_lower_bound_trees} and the proof of the lower bound in the random setting is given in \Cref{subsec:trees_random_lower_bound}.






\subsection{Upper Bounds}
\label{sec:dagup}

\ignore{
It is easy to see that for any rooted tree $G$, we have that
\begin{equation}\label{eq:treesone}
    \herdisc(\prefix(G)) \le \Delta(\prefix(G), 1) = \sup_{v_1, \dots, v_T \in [-1,1]} \pathdisc_G(v_1,\ldots,v_T) = 1,
\end{equation} 
by using a simple greedy procedure. Also, note that $\herdisc(\prefix(G))\ge 1$. 
}

The proof of \Cref{thm:DAGs_prefix} has two parts --- first generalizing Banaszczyk's result in \Cref{thm:Banaszczyk12} to prefixes of rooted trees (which can be naturally viewed as DAGs by orienting the edges from root to leaves), and second, reducing the case of an arbitrary DAG to a rooted tree.

For the first part, we prove the following for input vectors in $\R^d$.


\begin{lemma}[Prefix Discrepancy for Trees] \label{lem:prefix_trees}
Let $\CT$ be a given $T$ node rooted tree and dimension $d\ge 2$. For any $v_1, \ldots, v_T \in \B_2^d$, we have that $\pathdisc_\CT(v_1, \ldots, v_T) = O(\sqrt{\log dT})$. 
More generally, given any convex body $K \in \R^d$ with Gaussian measure $\gauss(K) \ge 1 - 1/(2T)$, there exists a signing $x \in \pmone^T$ such that for every $S \in \prefix(\CT)$, we have that $\sum_{t \in S} x_t v_t \in O(1) \cdot K$. 
\end{lemma}

\begin{remark}\label{rem:done-case}
   For the $d=1$ case, it is easy to see that for any rooted tree $\CT$, we have
\begin{equation}\label{eq:treesone}
    \herdisc(\prefix(G)) \le \Delta(\prefix(G), 1) = \sup_{v_1, \dots, v_T \in [-1,1]} \pathdisc_G(v_1,\ldots,v_T) = 1,
\end{equation} 
by using a simple greedy procedure. Also, note that $\herdisc(\prefix(\CT))\ge 1$. 
\end{remark} 

For the second part, let $G = ([T],E)$ be an arbitrary DAG and recall that $\prefix(G)$ denotes the set of all paths in $G$ starting from the topologically ordered root. We show that one can remove edges from $G$ to obtain a tree $\calT \subseteq G$ such that every path in $\prefix(G)$ contains at most $O(\herdisc(\prefix(G)))$ non-tree edges. 

\DAGreduction*

The bound in \Cref{lem:reduction_DAGs_to_trees} is the best possible one can hope for in terms of hereditary discrepancy. In particular, for any tree $\CT$ on the same vertex set $[T]$, there exists a directed path $\mathcal{P} \in \prefix(G)$ such that $\CP$ has at least $\herdisc(\prefix(G))/2$ non-tree edges. 
To see this, we assume, for the purpose of contradiction, that there is a tree $\CT$ with vertex set $[T]$ such that any path $\mathcal{P} \in \prefix(G)$ contains strictly less than $\herdisc(\prefix(G))/2$ non-tree edges. Then any path $\mathcal{P}\in \prefix(G)$ can be written as the disjoint union of less than $\herdisc(\prefix(G))/2$ paths whose edges lie entirely in $\CT$.
Now, restricted to any subset of $[T]$ as in Definition~\ref{def:herdisc}, since $\herdisc(\prefix(\CT)) = 1$, there is a signing such that every path (not necessarily starting from the root) in $\CT$ has discrepancy at most $2$. The discrepancy along any path $\mathcal{P} \in \prefix(G)$ is thus strictly smaller than $\herdisc(\prefix(G))$, contradicting the definition of hereditary discrepancy.


Therefore, \Cref{lem:reduction_DAGs_to_trees} also gives a combinatorial characterization (up to a constant factor) of the hereditary discrepancy for $\prefix(G)$ as follows.

\begin{corollary}[Combinatorial Characterization of $\herdisc(\prefix(G))$] \label{cor:comb_char_DAG} For any DAG $G=([T],E)$, we have
\begin{align*}
    \herdisc(\prefix(G)) = \Theta(1) \cdot \min_{\text{tree }\CT } \max_{\mathcal{P} \in \prefix(G)} |\mathcal{P} \cap (G - \CT)|,
\end{align*} 
where the minimum ranges over all trees $\CT$ on the vertex set $[T]$. 
\end{corollary}

\Cref{thm:DAGs_prefix} then follows immediately by combining \Cref{lem:reduction_DAGs_to_trees} with the prefix discrepancy bound for trees given in \Cref{lem:prefix_trees} and \cref{eq:treesone}.

\begin{proof}[Proof of \Cref{thm:DAGs_prefix}]
By \Cref{lem:reduction_DAGs_to_trees}, 
there exists a (rooted) tree $\CT \subseteq G$ with the same set of vertices such that for any directed path $\mathcal{P}$ in $G$, the number of non-tree edges is $O(\herdisc(\prefix(G)))$. We then use \Cref{lem:prefix_trees} to find a coloring $x \in \{-1,1\}^{T}$ for the vectors $\{v_t\}_{t \in [T]}$ such that each path in $\CT$ has discrepancy at most $O(\sqrt{\log d + \log T})$. Now, any directed path $\mathcal{P}$ of the DAG $G$ can be decomposed into at most  $O(\herdisc(\prefix(G)))$ disjoint sub-paths in $\CT$, and each of these sub-paths has discrepancy $O(\sqrt{\log d+ \log T})$ under coloring $x$ (since the discrepancy of any sub-path can be written as a difference of discrepancy of two prefixes). It thus follows that the discrepancy of $\mathcal{P}$ under coloring $x$ is at most $O(\herdisc(\prefix(G)) \cdot \sqrt{\log d + \log T})$. 

The second statement also follows immediately by using Eq.~\eqref{eq:treesone} instead of \Cref{lem:prefix_trees} above. This completes the proof of \Cref{thm:DAGs_prefix}. 
\end{proof}

\subsubsection{Rooted Trees}
\label{subsec:trees_prefix}

In this section, we prove \Cref{lem:prefix_trees}.

\begin{proof}[Proof of \Cref{lem:prefix_trees}]
We note that the more general statement regarding convex bodies directly implies $\pathdisc_\CT(v_1, \ldots, v_T) = O(\sqrt{\log dT})$ for any given vectors $v_1, \cdots, v_T \in \B_2^d$. This follows by taking the convex body $K$ to be $O(\sqrt{\log dT}) \cdot \B_\infty^d$ for which the  Gaussian measure $\gauss(K) \ge 1 - 1/(2T)$. 

Next, we shall prove the general statement for any arbitrary convex body $K$ with $\gauss(K) = 1-\delta$ where $\delta \le 1/(2T)$. For this, let us define a convex body $K_t$ for each node $t \in [T]$ from bottom up as follows. 
For any leaf $t$, we put $K_t := K$. 
For any intermediate node $t$, let $j_1, \cdots, j_\ell$ be its children and define 
\begin{align*}
    K_t := K \cap \left(\cap_{k \in [\ell]} (K_{j_k} * v_{j_k}) \right) ,
\end{align*}
where $K_{j_k} * v_{j_k} \subseteq (K_{j_k} + v_{j_k}) \cup (K_{j_k} - v_{j_k})$ is the symmetrization given in \Cref{defn:ban_symmetrization}. 
Note that for any vector $w \in K_t$ (think of $w$ as the discrepancy vector from the root to node $t$), there exist signs $x_{j_1}, \cdots, x_{j_\ell} \in \pmone^\ell$ such that $w + x_{j_s}v_{j_s} \in K_{j_s}$, for any $s \in [\ell]$.
In other words, once we can find signs for the path from the root $r$ to the node $t$ such that the discrepancy vector lies in $K_t$, we can also inductively find signs for the children $j_1, \cdots, j_\ell$ such that the discrepancy vector from the root $r$ to to these nodes lie in $K_{j_1}, \cdots, K_{j_\ell}$ respectively. 
Recall from \Cref{lem:ban_symmetrization} that $\gauss(K_{j_k} * v_{j_k}) \geq \gauss(K_{j_k})$ whenever $\gauss(K_{j_k}) \geq 1/2$ and $\| v_{j_k}\|_2 \leq 1/5$.

For each intermediate node $t \in \CT$, we let $n_t$ be the number of nodes in the subtree rooted at $t$. 
We prove inductively that 
\begin{align} \label{eq:induction_hypothesis}
    \gauss(K_t) \geq 1 - n_t \delta  \enspace . 
\end{align}

For every leaf node $t$, we have $n_t = 1$ so \eqref{eq:induction_hypothesis} follows from the assumption that $\gauss(K) \geq 1-\delta$. 
Let $t$ be an intermediate node with children $j_1, \cdots, j_\ell$ and suppose \eqref{eq:induction_hypothesis} is established for each child of $t$. 
We thus have that $\gauss(K_{j_k}) \geq 1 - n_{j_k} \delta$ for each $k \in [\ell]$. 
Thus by union bound, we have
\begin{align*}
\gauss(K_t) ~\geq~ 1 - (\delta + \sum_{k \in [\ell]} n_{j_k} \delta) ~=~ 1 - n_t \delta  \enspace . 
\end{align*}
This proves the induction hypothesis \eqref{eq:induction_hypothesis}. 
In particular, the root node $r$ has $\gauss(K_r) \geq 1/2$, thus one can scale the length of vectors down by a constant factor so that $v_r \in K_r$. 
The rest of the signs can be found from the root to the leaves. 
\end{proof}



\newcommand{\tv}{m_v}
\newcommand{\ta}{m_a}
\newcommand{\tb}{m_b}
\newcommand{\tr}{m_r}
\newcommand{\tvarg}[1]{m_{#1}}

\subsubsection{General DAGs}

\label{subsec:DAGs_prefix}


In this subsection, we give the reduction from DAGs to trees and prove \Cref{lem:reduction_DAGs_to_trees}. We will need the following definition.

\begin{definition}[Number of Non-tree Edges]
Given a DAG $G$, let $\CT \subseteq G$ be a tree (or a forest). For any node $v \in G$, we define $\tv(\CT)$ to be the maximum number of edges on any path starting from node $v$ that are are not in the tree $\CT$. 
\end{definition}

In particular, if $r$ is the root of the DAG, then $\tr(\CT)$ is the maximum number of non-tree edges (with respect to $\CT$) on any path in $\prefix(G)$. Our goal would be to find a tree $\CT \subseteq G$ so that $\tr(\CT)$ is upper bounded by a combinatorial measure that we show gives a lower bound on hereditary discrepancy for prefixes of DAGs. For this we define: 
\begin{definition}[Chain structure]
\label{defn:chain_structure}
We say a DAG $C$ is a chain structure if it contains the following: 
there exist nodes $a_1, a_2, \cdots, a_\ell$ such that for any $i \in [\ell-1]$, there are two vertex disjoint paths $\mathcal{P}_i^+$ and $\mathcal{P}_i^-$ from $a_i$ to $a_{i+1}$ such that $|\mathcal{P}_i^+|$ contains at least one node other than $a_i$ and $a_{i+1}$. The length of the chain is defined as $\ell_C := \ell-1$. 
\end{definition}

See \Cref{fig:chain} for an example. The length of a chain structure gives a lower bound for $\herdisc(\prefix(C))$:
\begin{claim}
For any chain structure $C$, we have $\herdisc(\prefix(C)) \geq \ell_C/4$.
\end{claim}
\begin{proof}
We find a subset of vertices $J$ for which the discrepancy of the set system $\prefix(C)$ projected on $J$ is lower bounded by the above: we add the vertices $a_1, \cdots, a_\ell$ to $J$, and add exactly one intermediate node $b_i$ from each $\mathcal{P}_i^+$ to $J$. 

Now, consider any coloring $x_i \in \{-1,1\}^\ell$ for the $a_i$'s and $y_i \in \{-1,1\}^{\ell-1}$ for the $b_i$'s.
We may assume that $|\sum_{i \in [\ell]} x_i | < \ell_C/4$, as otherwise the path from $a_1$ to $a_\ell$ along all the paths $\mathcal{P}_i^-$ already has discrepancy at least $\ell_C/4$. 
Now there exists a color $\eps \in \pmone$ such that $\sum_{i \in [\ell-1]: y_i = \eps} y_i \geq \ell_C/2$. Then the path from $a_1$ to $a_\ell$ that passes through all the $b_i$'s with color $\eps$ and otherwise takes $\mathcal{P}_i^-$ has discrepancy at least $\ell_C/2 - \ell_C/4 = \ell_C/4$ when projected to $J$. This proves the claim. 
\end{proof}

We now prove \Cref{lem:reduction_DAGs_to_trees} with the above setup.

\begin{proof}[Proof of \Cref{lem:reduction_DAGs_to_trees}]
Let $\ell_r$ be the length of the longest chain that is a subgraph of the DAG $G$ with root $r$. Then, it suffices to show that there is a tree $\CT \subseteq G$ such that $m_r \le \ell_r \le 4\herdisc(\prefix(G))$. We prove this inductively and for the induction, we define $\ell_v$ to be the length of the longest chain in the sub-DAG of $G$ rooted at $v$. We will show the following claim which immediately implies the statement of \Cref{lem:reduction_DAGs_to_trees} by taking $v$ to be the root $r$.

\begin{claim}
\label{lem:tree_removed_edges}
Let $G$ be a DAG. Then, there exists a tree $\CT \subseteq G$ such that $\tv(\CT) \leq \ell_v$ for any node $v \in G$. 
\end{claim}

\ignore{
\begin{proof}[Proof of \Cref{claim:remove_free}]
Consider any node $v \in G$ and let $\mathcal{P}_v$ be a path starting from $v$ with $e \notin \mathcal{P}_v$ that contains the maximum number of edges in $G - \CT$, and $\mathcal{Q}_v$ be a path starting from $v$ with $e \in \mathcal{Q}_v$ that contains the maximum number of edges in $G - \CT$. Note that since $\ta(\CT) > \tb(\CT)$, we must have $t(\mathcal{P}_v, \CT) > t(\mathcal{Q}_v, \CT)$ because the path $\mathcal{Q}_v$ is forced to take the edge $e = (a,b)$ in $\CT$ which  decreases the $m$ value (otherwise, it can take a different path starting from node $a$ that contains $\ta(\CT)$ removed edges). 
For a path $\mathcal{P}$ in $G$, let's denote $m(\mathcal{P}, \CT) =  = |\mathcal{P} \cap (G-\CT)|$ to be the number of non-tree edges of $\mathcal{P}$ with respect to $\CT$. It follows that $\tv(\CT) = \max\{t(\mathcal{P}_v, \CT), t(\mathcal{Q}_v, \CT)\}$ and that $\tv(\CT-e) = \max\{t(\mathcal{P}_v, \CT), t(\mathcal{Q}_v, \CT)+ 1\}$. This implies that $\tv(\CT-e) = \tv(\CT)$ and finishes the proof of the claim. 
\end{proof}}

\begin{proof}[Proof of \Cref{lem:tree_removed_edges}] We start by giving an algorithm for building the tree $\CT$, and as the algorithm proceeds, we prove inductively that $\tv(\CT) \leq \ell_v$ for any node $v$ already considered by the algorithm.  

\medskip
\noindent \textbf{Tree building algorithm.}  
Start with $\CT = \emptyset$. 
We visit the nodes in the DAG in the reverse topological order. 
Suppose we are looking at a node $u$ and let $v_1, \cdots, v_k$ be its children. 
Let $G_u$ be the sub-DAG rooted at $u$ and $\CT_u$ be the edges already added to the tree by the algorithm (not including the edges $(u, v_1), \cdots, (u, v_k)$). The vertex set of $G_u$ and $\CT_u$ is exactly the set of vertices already visited by the algorithm (including node $u$ which we add to the vertex set of $\CT_u$).
For any node $v \in G_u$, we shall overload the notation $\tv(\CT_u)$ to denote the maximum number of edges in $G_u - \CT_u$ contained in any path starting from node $v$.
We assume as induction hypothesis that $\tv(\CT_u) \leq \ell_v$ for any node $v \in G_u - \{u\}$. 
Now we are going to add some of the edges in $(u,v_1), \cdots, (u, v_k)$ to the tree $\CT_u$ and possibly remove some edges from $\CT_u$ to obtain a new tree $\CT'_u$ such that $\tv(\CT'_u) \leq \ell_v$ for any node $v \in G_u$ (including the vertex $u$). 
If this could be done, then the algorithm can proceed inductively in the reverse topological order. 

Reorder the vertices $v_1, \cdots, v_k$ such that $\tvarg{v_1}(\CT_u) = \tvarg{v_2}(\CT_u) = \cdots = \tvarg{v_i}(\CT_u) > \tvarg{v_{i+1}}(\CT_u) \geq \cdots \geq \tvarg{v_k}(\CT_u)$. 
We add to tree $\CT_u$ one by one the edges in the order of $(u,v_1), \cdots, (u,v_i)$ (we are never going to add the edges $(u, v_{i'})$ for any $i' > i$ because they have smaller $t$ values and therefore don't matter).
If adding these edges to $\CT_u$ do not form a cycle, then we are done as the induction hypothesis immediately holds for $u$ in this case: $\tvarg{u}(\CT'_u) = \tvarg{v_1}(\CT'_u) \leq \ell_{v_1} \leq \ell_u$. 

Suppose when adding $(u,v_j)$, where $j \leq i$, we first encounter a cycle $W$ (in the undirected sense) formed by edges of $\CT_u$ and $(u,v_1), \cdots, (u,v_j)$.
The cycle $W$ must contain the edge $(u,v_j)$ as well as $(u,v_{j'})$ for some $j' < j$. Let node $a \in W$ be the vertex with the biggest topological order in the sub-graph given by the cycle $W$, i.e. the out-degree of $a$ in the cycle $W$ is $0$. Note that we must have $\ta(\CT_u) \leq \tvarg{v_j}(\CT_u)$ because $a$ is a successor of $v_j$ in $G_u$. 
We consider two cases:


\smallskip
\noindent \textbf{Case (1): $\ta(\CT_u) = \tvarg{v_j}(\CT_u)$.}  In this case, our algorithm can simply stop adding any of the edges in $(u,v_j), \cdots, (u,v_i)$. Denote the tree $\CT'_u$ as the tree $\CT_u$ by adding the edges $(u,v_1), \cdots, (u,v_{j-1})$. Note that $\tv(\CT'_u) = \tv(\CT_u)$ for any node $v \in G_u - \{u\}$, so we only need to prove that $\tvarg{u}(\CT'_u) \leq \ell_u$. 
To this end, we first note that $\tvarg{u}(\CT'_u) \leq \tvarg{v_j}(\CT'_u) + 1 = \tvarg{v_j}(\CT_u) + 1$. 
Now the key observation is that $\ell_u \geq \ell_a + 1$, as the cycle $W$ together with a chain of length $\ell_a$ starting from node $a$ forms a new chain with length $\ell_a + 1$. 
It then follows that 
\begin{align*}
    \ell_u ~\geq~ \ell_a + 1 ~\geq~ \ta(\CT_u) + 1 ~=~ \tvarg{v_j}(\CT_u) + 1 ~\geq~ \tvarg{u}(\CT'_u) \enspace . 
\end{align*}

In other words, in this case we could afford to increase the $m$ value by one at the node $u$ as the length of the chain at $u$ also increases.

\smallskip
\noindent \textbf{Case (2): $\ta(\CT_u) < \tvarg{v_j}(\CT_u)$.}
In this case, there must exist an edge $(b_1, b_2) \in W$ such that $\tvarg{b_1}(\CT_u) > \tvarg{b_2}(\CT_u)$. We claim that we can remove edge $(b_1,b_2)$ from $\CT_u$ without changing $\tv(\CT_u)$ for any node $v \in \CT_u$. 


\begin{proposition}[Removing an Edge for Free] \label{claim:remove_free}
Let $G$ be a DAG and $\CT \subseteq G$. Suppose edge $e = (a, b) \in \CT$ and $\ta(\CT) > \tb(\CT)$, then $\tv(\CT-e) = \tv(\CT)$ for every node $v \in G$. 
\end{proposition}

We defer the proof of the above proposition and continue with the proof of correctness of the algorithm. Note that after removing $(b_1,b_2)$, we can safely add edge $(u,v_j)$ to $\CT_u - (b_1,b_2)$ without introducing any cycle. 
Phrased differently, the occasion of $\ta(\CT_u) < \tvarg{v_j}(\CT_u)$ is an indication that we can further simplify $\CT_u$ by removing one edge without increasing the $m$ value of any nodes in it, and doing so avoids having a cycle after the edge $(u,v_j)$ is added. Note that adding any edge $(u,v_j)$ does not change $\tvarg{v_{j'}}(\CT_u)$ for any $j' \in [k]$, since node $v_{j'}$ has a bigger topological order than $u$ and thus the DAG property implies that any path starting from node $v_{j'}$ would never pass through the node $u$.  
If we are always in the case (2), we eventually reach a tree $\CT'_u$ such that $\tvarg{u}(\CT'_u) = \tvarg{v_i}(\CT'_u) = \tvarg{v_i}(\CT_u)$.
As we have $\ell_u \geq \ell_{v_i} \geq \tvarg{v_i}(\CT_u)$, we conclude that $\ell_u \geq \tvarg{u}(\CT'_u)$. 

Taking $\CT'_u$ to be the the new tree $\CT$, we can inductively continue the tree building algorithm since we have now satisfied the invariant for the sub-DAG rooted at $u$. This finishes the proof of \Cref{lem:tree_removed_edges}. \qedhere


\end{proof}

 We now prove \Cref{claim:remove_free}.
\begin{proof}[Proof of \Cref{claim:remove_free}]
Consider any node $v \in G$ and let $\mathcal{P}_v$ be a path starting from $v$ with $e \notin \mathcal{P}_v$ that contains the maximum number of edges in $G - \CT$, and $\mathcal{Q}_v$ be a path starting from $v$ with $e \in \mathcal{Q}_v$ that contains the maximum number of edges in $G - \CT$ (set $\mathcal{Q}_v = \emptyset$ if such a path does not exist).  For a path $\mathcal{P}$ in $G$, let us denote $m(\mathcal{P}, \CT)  = |\mathcal{P} \cap (G-\CT)|$ to be the number of non-tree edges of $\mathcal{P}$ with respect to $\CT$. It follows that $\tv(\CT) = \max\{m(\mathcal{P}_v, \CT), m(\mathcal{Q}_v, \CT)\}$ and that $\tv(\CT-e) = \max\{m(\mathcal{P}_v, \CT), m(\mathcal{Q}_v, \CT)+ 1\}$. 
In the remainder of the proof, we show that $m(\mathcal{P}_v, \CT) > m(\mathcal{Q}_v, \CT)$, which would immediately imply $\tv(\CT-e) = \tv(\CT)$. 
The intuition here is that since $\ta(\CT) > \tb(\CT)$, the path $\mathcal{Q}_v$ is forced to take the edge $e = (a,b)$ in $\CT$ which  decreases the $m$ value (otherwise, it can take a different path starting from node $a$ that contains $\ta(\CT)$ removed edges). 

To be more precise, we first prove that $m(\mathcal{P}_a, \CT) > m(\mathcal{Q}_a, \CT)$. Note that $m(\mathcal{Q}_a, \CT) = m(\mathcal{P}_b, \CT) = \tb(\CT) < \ta(\CT)$. Since $\ta(\CT) =  \max\{m(\mathcal{P}_a, \CT), m(\mathcal{Q}_a, \CT)\}$, we have $m(\mathcal{P}_a, \CT) = \ta(\CT) > m(\mathcal{Q}_a, \CT)$. 
Now for any vertex $v \in G$, let $\mathcal{Q}_{v \rightarrow a}$ be the sub-path of $\mathcal{Q}_v$ from $v$ to $a$.
Consider the path $\mathcal{Q}' := \mathcal{Q}_{v \rightarrow a} \cup \mathcal{P}_a$ starting from vertex $v$, we then have 
\begin{align*}
    m(\mathcal{Q}', \CT) 
    &~=~ m(\mathcal{Q}_{v \rightarrow a}, \CT) + m(\mathcal{P}_a, \CT) ~~>~~ m(\mathcal{Q}_{v \rightarrow a}, \CT) + m(\mathcal{Q}_a, \CT) \\
    & ~\geq~ m(\mathcal{Q}_{v \rightarrow a}, \CT) + m(\mathcal{Q}_v - \mathcal{Q}_{v \rightarrow a}, \CT) ~~=~~ m(\mathcal{Q}_v, \CT) \enspace .
\end{align*}
It then follows that  $m(\mathcal{P}_v, \CT) \geq m(\mathcal{Q}', \CT) > m(\mathcal{Q}_v, \CT)$. 
This completes the proof of the proposition.
\end{proof}
This finishes the proof of \Cref{lem:tree_removed_edges} and also \Cref{lem:reduction_DAGs_to_trees}.
\end{proof}

\subsection{Lower Bounds}
\label{sec:dag-lb}

For the proof of the lower bounds for prefixes in the case of rooted trees, we will need the following definition.

\begin{definition}[Canonical Coloring for Vertices of a Complete Binary Tree]
\label{defn:canonical_coloring}
Let $\mathcal{T}=([T],E)$ be a complete binary tree with root $r$. For each node $t \in \mathcal{T}$, let $\mathcal{P}_t \subseteq [T]$ be the unique path from $r$ to $t$.
We define the canonical coloring for a vertex $t$ as the coloring $y_t \in \pmone^{\mathcal{P}_t \setminus \{t\}}$ which assigns a node $j \in \mathcal{P}_t \setminus \{t\}$ the color $x_t^+(j) = -1$ if $\mathcal{P}_t$ goes to $j$'s left child, and $x_t^+(j) = +1$ otherwise. 
\end{definition}

\ignore{
\begin{definition}[Canonical coloring]
\label{defn:canonical_coloring}
Let $\mathcal{T} = ([T],E)$ be a complete binary tree with root $r$.
For each node $t \in \mathcal{T}$, let $\mathcal{P}_t \subseteq [T]$ be the unique path from $r$ to $t$.
We define the positive canonical coloring $x_t^+ \in \pmone^{\mathcal{P}_t}$ for $\mathcal{P}_t$ as the one which assigns color $x_t^+(t) = +1$ to $t$, and for any $j \in \mathcal{P}_t \setminus \{t\}$ assigns color $x_t^+(j) = -1$ if $\mathcal{P}_t$ goes to $j$'s left child, and $x_t^+(j) = +1$ otherwise. 
The negative canonical coloring $x_t^- \in \pmone^{\mathcal{P}_t}$ is defined almost identically as $x_t^+$ except that $x_t^-(t) = -1$. 
\end{definition}}

We say that any given coloring $x \in \pmone^T$ agrees with the canonical coloring $y_u$ for the vertex $u$ if $x$ assigns the same sign as $y_u$ to each node of $\mathcal{P}_u \setminus \{u\}$. Furthermore, for any given coloring $x \in \pmone^{T}$, there always exists a leaf node $u$ in the tree such that $x$ agrees with the canonical coloring $y_u$ for the vertex $u$. To find such a leaf $u$, one simply follows the coloring $x$ from the root to the leaves, moving to the left child if $x$ colors the node $-1$, and the right child otherwise. 

\subsubsection{Adversarial Lower Bound for Trees}

\label{subsec:prefix_lower_bound_trees}

In this subsection, we prove \Cref{thm:prefix_lower_bound_trees}. 

\begin{proof}[Proof of \Cref{thm:prefix_lower_bound_trees}]
Without loss of generality, we assume that $T = 2^h-1$. Consider a $T$-node complete binary tree $\mathcal{T}$ with depth $h$. We construct the vectors $v_t \in \B_2^2$ at each node $t \in \mathcal{T}$ inductively in the following way. We start by picking an arbitrary unit vector $v_r $ for the root $r$. 
Then for any node $t \in \mathcal{T}$, let $j^-$ and $j^+$ be the left and right children of $t$. 
Let $d_t^+ = \sum_{i \in \mathcal{P}_t \setminus \{t\}} y_t(i)v_i + v_t$ and $d_t^- = \sum_{i \in \mathcal{P}_t \setminus \{t\}} y_t(i)v_i - v_t$ where $y_t$ is the canonical coloring for vertex $t$. The vectors $v_{j^-}$ and $v_{j^+}$ are then picked as arbitrary unit vectors orthogonal to $d_t^-$ and $d_t^+$ respectively. 

Let $x \in \pmone^T$ be a given coloring and $u$ be the leaf node such that $x$ agrees with the canonical coloring $y_u$ for the vertex $u$. Then, on the path $\mathcal{P}_u \setminus \{u\}$, every vector is orthogonal to the previous discrepancy vector under coloring $x$, and thus the discrepancy vector along $\mathcal{P}_u \setminus \{u\}$ has $\ell_2$ norm at least $\sqrt{h-1}$. Since the vectors are $2$-dimensional, the $\ell_\infty$ norm of the discrepancy vector along $\mathcal{P}_u \setminus \{u\}$ is at least $\Omega(\sqrt{h}) = \Omega(\sqrt{\log T})$. 
This completes the proof of \Cref{thm:prefix_lower_bound_trees}.
\end{proof}


\subsubsection{Stochastic Lower Bound for Trees}
\label{subsec:trees_random_lower_bound}

In this subsection, we prove \Cref{thm:trees_random_lower_bound}. 
The idea is to show that the lower bound example in \Cref{thm:prefix_lower_bound_trees} (with approximate orthogonality to the canonical discrepancy vectors) appears with high probability in an $\ell$-ary tree, where $\ell = \poly(\log T)$ to be specified later and the vectors at the vertices are sampled i.i.d. from the unit sphere in $\R^2$.

\begin{proof}[Proof of \Cref{thm:trees_random_lower_bound}]
We consider a $T$-node complete $\ell$-ary tree $\mathcal{T}'$ with depth $h$ and root $r$. Note that $\ell$ and $h$ satisfies that $\sum_{i=0}^{h-1} \ell^i = \ell^{\Theta(h)} = T$. The vectors $v_t \in \R^2$ for each node $t$ of the tree are picked i.i.d. from $\S^1$. For each node $t \in \mathcal{T}$, let $\mathcal{P}_t \subseteq [T]$ be the unique path from $r$ to $t$ and let $d_t^+ = \sum_{i \in \mathcal{P}_t \setminus \{t\}} y_t(i)v_i + v_t$ and $d_t^- = \sum_{i \in \mathcal{P}_t \setminus \{t\}} y_t(i)v_i - v_t$ where $y_t$ is the canonical coloring for vertex $t$.

Our goal is to show that w.h.p. there exists a complete binary subtree $\mathcal{T}' \subseteq \mathcal{T}$ with depth $h$ such that for any non-leaf node $t \in \mathcal{T}'$, its left (resp. right) child $j^-$ (resp. $j^+$) and negative (resp. positive) canonical discrepancy vector $d_t^-$ (resp. $d_t^+$) satisfies $|\langle v_{j^-}, d_t^- \rangle| \leq 1/4$ (resp. $|\langle v_{j^+}, d_t^+ \rangle| \leq 1/4$). If the above event occurs, then the discrepancy of the prefixes of the tree must be $\Omega(\sqrt{h})$. To see this, take any coloring $x \in \{-1,1\}^{|\mathcal{T}'|}$, and consider the leaf $u \in \mathcal{T}'$ such that $x$ agrees with the canonical coloring for the leaf $u$. 
Since $v_{j^-}$ (resp. $v_{j^+}$) is almost orthogonal to $d_t^-$ (resp. $d_t^+$) for any internal node $t$ on the path to $u$, the square of the $\ell_2$-discrepancy for the coloring $x$ along the root-leaf path $\mathcal{P}_u$ increases by at least $\Omega(1)$ after each node (except maybe for the leaf $u$). As the dimension $d=2$, the $\ell_\infty$-discrepancy for the coloring $x$ along the root-leaf path $\mathcal{P}_u$ is also $\Omega(\sqrt{h})$.

It thus remains to show that such a complete binary tree $\mathcal{T}'$ exists w.h.p. We find such a tree $\mathcal{T}'$ inductively and start by adding the root $r$ to $\mathcal{T}'$. For any node $t$ that was just added to $\mathcal{T}'$, we hope to find left and right children $j^-$ and $j^+$ among the first and second $\ell/2$ children of $t$ such that both $|\langle v_{j^-}, d_t^- \rangle| \leq 1/4$ and $|\langle v_{j^+}, d_t^+ \rangle| \leq 1/4$ hold.
For any child $j$ of $t$, since $\| d_t^- \|_2 \leq \sqrt{h}$, we have
\begin{align*}
    \p_{v_j \sim \S^1} \Big[|\langle v_j, d_t^- \rangle| \leq 1/4 \Big] ~\geq~ \frac{1}{10 \sqrt{h}}  \enspace . 
\end{align*}
Thus the probability that none of the first $\ell/2$ children satisfies this property is thus upper bounded by
\begin{align*}
(1 - 1/(10\sqrt{h}))^{\ell/2} ~\leq~ \exp(- \Theta(\ell / \sqrt{h})) \enspace . 
\end{align*}
Thus taking $\ell = C \cdot \sqrt{h} \cdot \log T$ for a sufficiently large constant $C$ suffices to guarantee that the above probability is at most $1/T^2$. 
One could therefore afford to do a union bound over all the nodes in the process of building the complete binary tree $\mathcal{T}'$. 
To satisfy $\ell^{\Theta(h)} = T$, we need to set $h = \log T/\log \log T$. 
Thus the prefix discrepancy on $\mathcal{T}'$ is thus at least $\Omega(\sqrt{h}) = \Omega(\sqrt{\log T / \log \log T})$. 
This completes the proof of  \Cref{thm:trees_random_lower_bound}. 
\end{proof}


\section{Combinatorial Vector Balancing}
\label{sec:comb_vec_bal}




In this section we prove our result for combinatorial vector balancing --- given a set family $\CS$ on the  ground set $[T]$ and input vectors $v_1, \ldots, v_T \in \B_2^d$,  the goal is to find a signing $x \in \pmone^T$ that minimizes
\[ \cdisc_\CS(v_1, \cdots, v_T) := \min_{x \in \{-1,+1\}^n} \max_{S \in \mathcal{S}} \Big\| \sum_{i \in S} x_i v_i \Big\|_\infty.\]

\combvecbal* 


To prove \Cref{thm:general-setsystem}, we bound $\cdisc_\CS(v_1, \cdots, v_T)$ in terms of $\gamma_2(A_\mathcal{S})$, where $A_\mathcal{S}\in \{0,1\}^{|\mathcal{S}| \times T}$ is the incidence matrix of the set system $\mathcal{S}$ (meaning that $(A_{S})_i = 1$ iff $i \in S$). \Cref{thm:general-setsystem} then immediately follows by using this bound and relating $\gamma_2$ norm to hereditary discrepancy via \Cref{lemm:gamma2_herdisc}.


\begin{lemma}
\label{lem:gamma2_upper_bound}
In the setting of \Cref{thm:general-setsystem}, we have that $\cdisc_\CS(v_1, \cdots, v_T) \leq \gamma_2(A_\mathcal{S}) \cdot \sqrt{\log d + \log |\mathcal{S}|} $. Moreover, the upper bound is constructive in the sense that there is a polynomial time algorithm to  find such a coloring. 
\end{lemma}

\begin{proof}[Proof of \Cref{lem:gamma2_upper_bound}]
Let the optimal factorization for $\gamma_2(A_{\mathcal{S}})$ be $A_\CS = BC$ where $\gamma_2(A_{\mathcal{S}}) = r(B) \cdot c(C)$ with $r(B)$ being the maximum $\ell_2$ norm among all rows of $B$ and $c(C)$ being the maximum $\ell_2$ norm among all columns of $C$. 

Denote the $j$th coordinate of vector $v_t \in \R^d$ as $v_t(j)$.
For each $j \in [d]$, we define  a matrix $D_j \in \R^{T \times T}$ as $D_j := \diag(v_1(j), \cdots, v_T(j))$.
Note that 
\begin{align*}
    \cdisc_\CS(v_1, \ldots, v_T) ~= ~ \min_{x \in \pmone^T} \max_{S \in \mathcal{S}, j \in [d]} \Big | \sum_{t \in S} x_t v_t(j) \Big | 
    ~= ~ \disc(A_\mathcal{S}^{D})  \enspace ,
\end{align*}
where the right hand side is the standard notion of discrepancy of a matrix\footnote{For a matrix $A \in \R^{d\times T}$, its discrepancy is defined as $\disc(A) = \min_{x\in \pmone^T} \|Ax\|_{\infty}$.}and the matrix $A_\mathcal{S}^{D} \in \R^{d|\mathcal{S}| \times T}$ is defined as
\begin{align*}
A_\mathcal{S}^{D} = 
\left( 
\begin{matrix}
A_\mathcal{S} D_1 \\
A_\mathcal{S} D_2 \\
\vdots\\
A_\mathcal{S} D_d 
\end{matrix}
\right) .
\end{align*}
Now a factorization of $A_\mathcal{S}^D$ is given by
\begin{align*}
    A_\mathcal{S}^D = 
    \left( 
    \begin{matrix}
    B &  &  & \\
      & B&  & \\
      &  & \ddots & \\
    & & & B 
    \end{matrix}
    \right)  \cdot
    \left( 
    \begin{matrix}
    C D_1 \\
    C D_2 \\
    \vdots \\
    C D_d 
    \end{matrix}
    \right)  \enspace .
\end{align*}
Note that the maximum $\ell_2$-norm of the rows of the first matrix is the same as $r(B)$, and the maximum $\ell_2$-norm of the columns of the second matrix is at most $c(C)$, as $\|v_i\|_2 \leq 1$ for each vector $i \in [T]$. 
It follows that $\gamma_2(A_\mathcal{S}^D) \leq r(B) \cdot c(C) = \gamma_2(A_\mathcal{S})$. 
The lemma then follows immediately from the second inequality in \Cref{lemm:gamma2_herdisc}. 
\end{proof}




\section{Open Problems and Future Directions} \label{sec:open}

We conclude by posing some open problems and  directions for future work. 

\paragraph{Strong  Koml\'os Conjecture.} 
The first conjecture concerns whether a stronger version of Koml\'os conjecture holds for prefix discrepancy.

\begin{open}\label{open:prefix}
    Given $T$ adversarial input vectors $v_1, \ldots, v_T \in \B_2^d$, does there always exist a signing $x \in \pmone^T$ such that $\predisc(v_1, \cdots, v_T) = O(1)$?
\end{open}

The $O(\sqrt{\log d + \log T})$ bound from \Cref{thm:Banaszczyk12} is also the best known result for the above problem. As pointed out before, our \Cref{thm:smoothed_upper_bound} implies that if there is a lower bound instance for the above problem matching Banaszczyk's bound, then such an instance must be sensitive to small perturbations.

Related questions  about prefix discrepancy  have also been posed before  with respect to different norms. 
\begin{open}[\cite{Spencer86Prefix}]\label{open:spencer}
    Given $T$ adversarial input vectors $v_1, \ldots, v_T \in \B_\infty^d$, does there always exist a signing $x \in \pmone^T$ such that $\max_{\tau \in [T]}  \big\|\sum_{t \le \tau} {x_t v_t} \big\|_\infty = O(\sqrt{d})$?
\end{open}

\begin{conjecture}[\cite{B12}]\label{open:ban}
     Given $T$ adversarial input vectors $v_1, \ldots, v_T \in \B_2^d$, does there always exist a signing $x \in \pmone^T$ such that $\max_{\tau \in [T]}  \big\|\sum_{t \le \tau} {x_t v_t} \big\|_2 = O(\sqrt{d})$?
\end{conjecture}
 Observe that an affirmative answer to \Cref{open:prefix} would give a positive answer to both these questions.
 %
 For $T=O(d)$, Spencer  \cite{Spencer86Prefix} showed a bound of $O(\sqrt{d})$ for \Cref{open:spencer}. The best known bound for \Cref{open:ban} is $O(\sqrt{d} + \sqrt{\log T})$ implied by \Cref{thm:Banaszczyk12}.

Another related open problem is regarding prefix discrepancy in the Beck-Fiala setting --- here, the input vectors are $s$ sparse and one is interested in a discrepancy bound in terms of the sparsity $s$. For the Beck-Fiala problem (where one is only interested in the final signed sum vector instead of all prefixes), the classical result of \cite{BeckFiala-DAM81} shows that the discrepancy is at most $O(s)$. Can we obtain an analogus result for all prefixes?
The best known results of \cite{BaranyGrinberg81, B12} imply a bound of $\min\{O(d),O(\sqrt{s\log(dT)})\}$. We note that a positive answer to \Cref{open:prefix} would imply an $O(\sqrt{s})$ bound on prefix discrepancy.

\paragraph{Smoothed Analysis and Combinatorial Vector Balancing.}
An interesting direction is to  improve the best known bounds for classical discrepancy problems  in the   smoothed analysis setting, e.g., the Koml\'os or Beck-Fiala conjectures. Since smoothed  settings captures stochastic settings, such improvements will be very powerful.

The next question concerns the combinatorial vector balancing problem introduced in this paper. We prove bounds in this paper that are tight up to polylogarithmic factor. But as the same bound holds for DAGs and paths (as in Banaszczyk's result), it is possible that it might hold in this more general setting as well. 

\begin{open}
\label{ques:relate-highd}
Can we show that for any general set-system $\mathcal{S}$ on $T$ elements, 
 \[
 \cdisc(\CS,d) \leq \herdisc(\mathcal{S})  \cdot O(\sqrt{\log d + \log T}) \enspace ?\]
\end{open}

\paragraph{Algorithmic Guarantees.} 
Lastly, it remains a very interesting problem in most of the cases above to make the corresponding results algorithmic, i.e., can we find a coloring achieving the bound in polynomial time. In  particular, an intriguing open question is to design an algorithm that achieves Banaszczyk's bound  for all prefixes given in  \Cref{thm:Banaszczyk12}.
If one only cares about the final discrepancy, then  \cite{BansalDGL18} gives an algorithm that matches this bound (also see~\cite{BansalG17} for a related algorithmic result). An algorithmic version of \Cref{thm:Banaszczyk12} would make several other discrepancy results algorithmic, including some of the results in this paper, the bound for Tusn\'ady's problem given in \cite{Nikolov-Mathematika19}, and the various applications of the Steinitz problem, as asked in \cite{BansalDGL18}. More generally, one can even hope for a polynomial-time algorithm for the  non-constructive result of  Banaszczyk for the prefix discrepancy setting~\cite[Lemma~2.5]{B12}, where one cares about convex bodies with large enough Gaussian measure. 

\ignore{
\begin{itemize}
    \item Strong Komlos conjectures
    
    \item Combinatorial Vec Balancing vs Hered discrepancy
    
    \item Smoothed analysis for well known discrepancy conjectures. Note that such smoothed guarantees are much stronger than those for stochastic inputs.
    
    \item Beck-Fiala for prefixes, can we get f(sparsity)?
    
    \item Make Banaczyk's prefix and Sasho's Tusnady result constructive.
\end{itemize}

General conjecture. Lower bound of $\sqrt{\log T}$ even for $d=2$. What about $d=1$?}

{\small
\bibliographystyle{alpha}
\bibliography{fullbib}
}
\end{document}